\documentclass[12pt]{amsart}
\usepackage[margin=1in]{geometry}
\usepackage{amsfonts}
\usepackage{amssymb}
\usepackage[dvips]{graphics}
\usepackage{epsfig}
\usepackage{euscript}
\usepackage{color}
\usepackage{tikz}
\usepackage{mathrsfs}
\usepackage [all,arc,curve,color, arrow, matrix,frame]{xy}

\usepackage[colorlinks,pdfstartview=FitB]{hyperref}
\hypersetup{allcolors=blue}
 \newtheorem{thm}{Theorem}[section]
 \newtheorem{cor}[thm]{Corollary}
 \newtheorem{lemma}[thm]{Lemma}
 \newtheorem{prop}[thm]{Proposition}
 \newtheorem{conjecture}[thm]{Conjecture}
 \theoremstyle{definition}
 \newtheorem{defn}[thm]{Definition}
 \theoremstyle{remark}
 \newtheorem{remark}[thm]{Remark}
 \newtheorem*{ex}{Example}
 
 \numberwithin{equation}{section}
\numberwithin{figure}{section}

\newcommand{\be}{\begin{equation}}
\newcommand{\ee}{\end{equation}}
\newcommand{\bea}{\begin{eqnarray}}
\newcommand{\eea}{\end{eqnarray}}

\newcommand{\C}{\mathbb{C}}
\newcommand{\R}{\mathbb{R}}

\newcommand{\Z}{\mathbb{Z}}

\newcommand{\E}{{\mathbb E}}

\newcommand{\N}{\mathbb{N}}

\DeclareMathOperator{\tr}{tr}

\DeclareMathOperator{\supp}{supp}

\def\idty{{\mathchoice {\mathrm{1\mskip-4mu l}} {\mathrm{1\mskip-4mu l}} %
{\mathrm{1\mskip-4.5mu l}} {\mathrm{1\mskip-5mu l}}}}

\newcommand{\Tr}{\operatorname{{\rm Tr}}}
\numberwithin{equation}{section}

\begin{document}

\title{Entanglement bounds in the XXZ quantum spin chain}

\author[H. Abdul-Rahman]{H.\ Abdul-Rahman$^1$}
\address{$^1$ Department of Mathematics\\
University of Arizona, Tucson, AZ 85721, USA}
\email{houssam@math.arizona.edu}

\author[C. Fischbacher]{C.\ Fischbacher$^2$}
\address{$^2$ Department of Mathematics\\
University of California, Irvine\\
Irvine, CA, 92697, USA}
\email{fischbac@uci.edu}

\author[G. Stolz]{G.\ Stolz$^3$}
\address{$^3$ Department of Mathematics\\
University of Alabama at Birmingham\\
Birmingham, AL, 35294, USA}
\email{stolz@uab.edu}

\date{}

%
\begin{abstract}

We consider the XXZ spin chain, characterized by an anisotropy parameter $\Delta>1$, and normalized such that the ground state energy is $0$ and the ground state given by the all-spins-up state. The energies $E_K = K(1-1/\Delta)$, $K=1,2,\ldots$, can be interpreted as $K$-cluster break-up thresholds for down-spin configurations. We show that, for every $K$, the bipartite entanglement of all states with energy below the $(K+1)$-cluster break-up satisfies a logarithmically corrected (or enhanced) area law. This generalizes a result by Beaud and Warzel, who considered energies in the droplet spectrum (i.e., below the 2-cluster break-up).

For general $K$, we find an upper logarithmic bound with pre-factor $2K-1$. We show that this constant is optimal in the Ising limit $\Delta=\infty$. Beaud and Warzel also showed that after introducing a random field and disorder averaging the enhanced area law becomes a strict area law, again for states in the droplet regime. For the Ising limit with random field, we show that this result does not extend beyond the droplet regime. Instead, we find states with energies of an arbitrarily small amount above the $K$-cluster break-up whose entanglement satisfies a logarithmically growing lower bound with pre-factor $K-1$.

\end{abstract}

\maketitle


\section{Introduction}

\subsection{Background}

The bipartite entanglement of quantum many-body systems has been the topic of intensive research, in particular in the physics literature, see the surveys \cite{Amicoetal,Eisertetal} and the articles in the special issue \cite{Calabreseetal}. While the entanglement of generic states satisfies a volume law with respect to the size of the chosen subsystem, it has been found that ground states in some of the standard examples have much lower entanglement. Hastings' celebrated result \cite{Hastings} says that gapped ground states in a large class of one-dimensional spin chains satisfy an area law, i.e., the entanglement doesn't grow with the subsystem. For the XY chain, where an explicit analysis is possible, either an area law (in non-critical cases) or a logarithmically corrected (or enhanced) area law (in some critical cases, in particular the isotropic XX chain) is found \cite{Vidaletal, JinKorepin}.

It has been argued that the validity of an area law for states in an energy regime extending beyond the ground state can be considered as a manifestation of many-body localization (MBL), e.g.\ \cite{BauerNayak,BrandaoHorodecky}. In this view, states with a log-corrected area law might still be considered as many-body localized, even if in slightly weaker form. Many of the obtained results on effects of disorder on entanglement for specific models are numerical (e.g.\ \cite{Bardarsonetal, Serbynetal} for a study of the growth of the dynamic entanglement in the XXZ chain) or use non-rigorous methods from quantum field theory (reviewed in \cite{CalabreseCardy}).

So far, mathematically rigorous results are mostly restricted to simple models whose study can be reduced to effective one particle Hamiltonians, e.g.\ disordered quantum oscillator systems (e.g.\ \cite{AR18, ARS, Beaudetal} and references therein) or the random XY spin chain (see the survey \cite{ARNSS}).

Interesting mathematical progress in understanding the effect of disorder on quantum many-body systems has recently been made for the XXZ spin chain in the Ising phase \cite{BW17, BW18, EKS, EKS2} (characterized by values $1<\Delta<\infty$ of the anisotropy parameter, see \eqref{pairinteract} below). These works show that this model shows various MBL characteristics in the droplet regime above its gapped ground state energy, a regime studied before in \cite{NS, NSS2007, FS14}. In particular, the work by Beaud and Warzel \cite{BW18} shows that for the deterministic version of the model {\it all} states with energies in the droplet spectrum (not just the eigenstates) satisfy a $\log$-corrected area law, while the introduction of disorder into the model eliminates the $\log$-term, leading to a strict area law in the disorder average.

In the work presented here we study properties of the XXZ chain in the Ising phase for states with energies above the droplet spectrum. This is characterized by a series of threshold energies $E_K$, linearly growing in $K\in \N$, above which (and below $E_{K+1}$) the droplets can break up into up to $K$-clusters (so that the droplet spectrum corresponds to $[E_1,E_2))$. We consider the XXZ model on the chain $[1,L]$ for large $L$ and the von Neumann entanglement entropy $\mathcal{E}$ of states with respect to the decomposition into the subsystems $[1,\ell]$ and $[\ell+1,L]$.

Our main result shows that, for the deterministic model and states with energies up to $E_K$, the leading term of the entanglement satisfies $\mathcal{E}/ \log \ell \lesssim 2K-1$ in the limit $\ell \to \infty$ (see Theorem~\ref{thm:main} in the next section for a more precise statement). Thus, for each $K$, a $\log$-corrected area law holds uniformly for all states with energies up to $E_K$, with constant linearly growing in $K$ (and thus in the total energy of the state).

In addition, we include two illustrative results on the entanglement for the Ising limit $\Delta=\infty$ of the model: First, In Theorem~\ref{thm:detasymp} we show that for this limiting case the constant $2K-1$ is optimal. Second, after adding a disordered local field in the Ising limit, we find $\E(\max \mathcal{E}) \gtrsim (K-1)\log \ell$ for the maximal entanglement of states with energies below $E_K$ ($\E$ denoting the disorder average). Thus, for $K>1$ the introduction of disorder no longer produces a strict area law, at least not when considering {\it all} states in this energy regime.

This is not to say that an area law might follow also for the higher bands if one {\it only} considers eigenstates. In fact, this is what happens in the Ising limit and we believe that this is also true for finite anisotropy $\Delta$ (meaning that the $\log$-correction is due to the built-up of entanglement via multiple different eigenstates in a given band). If this is true remains an open question. Progress towards its answer might come from work in preparation by Elgart and Klein \cite{ElgartKlein}, who have announced results on MBL in the 2-cluster band.

\subsection{Model and Results} \label{sec:results}

We consider the XXZ spin chain in the Ising phase, given on the finite interval $\Lambda =[1,L] := \{1,\ldots,L\}$ by the Hamiltonian
\begin{equation}
H_\Lambda(V) = \sum_{j=1}^{L-1} h_{j,j+1} + \sum_{j=1}^L V_j \mathcal{N}_j + \beta (\mathcal{N}_1 + \mathcal{N}_L).
\end{equation}
Here
\begin{itemize}
\item We normalize the next-neighbor interaction as
\begin{equation} \label{pairinteract}
h_{j,j+1} = \frac{1}{4}(\idty-\sigma_j^Z \sigma_{j+1}^Z) - \frac{1}{4\Delta}(\sigma_j^X \sigma_{j+1}^X + \sigma_j^Y \sigma_{j+1}^Y),
\end{equation}
where $\sigma^X,\sigma^Y,\sigma^Z$ are the standard Pauli matrices and we assume $1<\Delta \le \infty$ (Ising phase).
\item For the exterior field, we assume $V\geq 0$, i.e., $V_j \ge 0$ for all $j$, where the down-spin projection $\mathcal{N}= \begin{pmatrix} 0 & 0 \\ 0 & 1 \end{pmatrix}$ takes the role of the local particle number operator.
\item For definiteness, we choose $\beta = \frac{1}{2}(1-\frac{1}{\Delta})$ in the boundary term, the smallest value at which suitable positivity properties hold (a {\it droplet boundary condition}, compare \cite{NS,BW17,EKS}).
\end{itemize}

In this normalization $E_0=0$ is the non-degenerate ground state energy with all-spins-up ground state. It is separated by a spectral gap $1-\frac{1}{\Delta}$ from the rest of the spectrum, which features a sequence of threshold energies
\begin{equation} \label{Ksplit}
E_K := K(1-\frac{1}{\Delta}), \quad K=1,2,\ldots,
\end{equation}
interpreted as $K$-cluster break-up energies in the following sense: For a subset $X\subseteq [1,L]$ let $\phi_X$ be the up-down-spin product state with down-spins at the positions given by $X$. If $X$ has at least $K$ connected components (i.e., $K$ down-spin droplets), then $\langle \phi_X, H_L(V) \phi_X \rangle \ge E_K$.

It is therefore that $[1-\frac{1}{\Delta}, 2(1-\frac{1}{\Delta})) \cap \sigma(H_L(V))$ is referred to as the {\it droplet spectrum} of $H_L(V)$. In this regime a number of many-body localization properties have recently been proven in \cite{BW17,BW18,EKS,EKS2} for the case that the $V_j$ are i.i.d.\ random variables. In particular, Beaud and Warzel \cite{BW18} have considered bounds on the bipartite entanglement of states in the droplet spectrum, both for deterministic and random field $V$. Our first main result is an extension of the deterministic result of \cite{BW18} to higher energies.

Let $\psi$ be a normalized state in ${\mathcal H}_{\Lambda} = \bigotimes_{j=1}^L \C^2$ and $\rho_\psi = |\psi\rangle \langle \psi|$ the corresponding rank one projection (pure state). For $1\le \ell < L$ we consider the bipartite decomposition ${\mathcal H}_{\Lambda} = {\mathcal H}_{\Lambda_0} \otimes {\mathcal H}_{\Lambda_0^c}$, where $\Lambda_0 = [1,\ell]$. Let $\rho_{\Lambda_0,\psi} = \Tr_{\Lambda_0^c} \rho_\psi$ be the reduced state on ${\mathcal H}_{\Lambda_0}$ and define the (bipartite) \emph{von Neumann entanglement entropy} of $\rho_\psi$ with respect to this decomposition as
\begin{equation} \label{entangle}
{\mathcal E}(\rho_\psi) = {\mathcal S}(\rho_{\Lambda_0,\psi}),
\end{equation}
where ${\mathcal S}(\rho_{\Lambda_0,\psi}) = - \Tr\,\rho_{\Lambda_0,\psi} \log \rho_{\Lambda_0,\psi}$ is the von Neumann entropy of the reduced state.

We will use the standard notation $\chi_M(H)$ for the spectral projection of a self-adjoint operator $H$ onto a set $M$.

For an arbitrary $K\in \N$, we will consider the entanglement of states whose energy is separated by a ``safety distance'' $\delta>0$ from the $(K+1)$-cluster break-up $E_{K+1}$:

\begin{thm} \label{thm:main}
For every $K\in\N$, every $\delta>0$ and every field $V\ge 0$ it holds that
\begin{equation} \label{eq:main}
\limsup_{\ell\to\infty} \limsup_{L\to\infty} \frac{\sup_{\psi} {\mathcal E}(\rho_\psi)}{\log \ell} \le 2K-1.
\end{equation}
Here the supremum is taken over all $\psi \in R(\chi_{[0,E_{K+1}-\delta]}(H_\Lambda(V)))$ with $\|\psi\|=1$.
\end{thm}

In short (and slightly non-rigorously): States with energy strictly below the $(K+1)$-cluster break-up satisfy the enhanced area law ${\mathcal E}(\rho_\psi) \lesssim (2K-1) \log \ell$ as $\ell\to\infty$ (after first taking the infinite volume limit $L\to\infty$).

\vspace{.3cm}

In the last section of this paper we carry out a detailed analysis of the eigenstate entanglement in the Ising limit $\Delta=\infty$ of the XXZ chain, that is the Hamiltonian
\begin{equation} \label{Ising1}
H_\Lambda^\infty(V) = H_\Lambda^{\infty} + \sum_{j=1}^L V_j \mathcal{N}_j,
\end{equation}
with
\begin{equation} \label{Ising2}
H_\Lambda^{\infty} = \frac{1}{4} \sum_{j=1}^{L-1} (\idty - \sigma_j^Z \sigma_{j+1}^Z) + \frac{1}{2} (\mathcal{N}_1 + \mathcal{N}_L).
\end{equation}

This operator is trivial from the point of view of diagonalization: It is easy to check that the product states $\{\phi_X; X\subseteq \Lambda\}$ provide a complete set of eigenvectors,
\begin{equation} \label{diagH}
H_\Lambda^\infty(V) \phi_X = (cl(X)+ \sum_{j\in X} V_j) \phi_X.
\end{equation}
Here $cl(X)$ is the number of connected components of $X$ (in the sense of subintervals of integers). It is here where the need of the boundary condition $\frac{1}{2} (\mathcal{N}_1 + \mathcal{N}_L)$ in \eqref{Ising2} comes in, assuring that \eqref{diagH} also allows for sets $X$ which contain one or both boundary points $1, L$.

In particular, all the eigenstates $\phi_X$ in (\ref{diagH}) are product states and thus trivially entangled. However, we can make at least two observations with respect to entanglement which are of some interest:

\subsubsection{Free field case} $H_\Lambda^\infty$ has eigenvalues $K=0,1,2,\ldots$. For $K\ge 1$ and for large $L$ the corresponding eigenspaces become highly degenerate, namely
\begin{equation}
\dim R(\chi_{\{K\}}(H_\Lambda^\infty)) = |\{Y \subseteq [1,L]: cl(Y)=K\}|.
\end{equation}
Thus more highly entangled eigenstates arise as linear combinations of the product states $\phi_X$. By mostly elementary counting arguments we will prove

\begin{thm} \label{thm:detasymp}
For any $K\in \N$ we have
\begin{equation} \label{conjdet}
\lim_{\ell \to \infty} \lim_{L\to \infty} \frac{\sup_{\psi} \mathcal{E}(\rho_\psi)}{\log \ell} = 2K-1,
\end{equation}
the supremum taken over $\psi \in R(\chi_{[0,K]}(H_\Lambda^\infty))$, $\|\psi\|=1$.
\end{thm}

Thus, at least in the Ising limit $\Delta=\infty$ and free field, where the spectral projection onto $[0,E_{K+1}-\delta]$ becomes the projection onto $[0,K]$, the upper bound in Theorem~\ref{thm:main} is optimal. One can indeed construct states with energy not larger than $K$ whose entanglement grows like $(2K-1) \log \ell$ for large $\ell$ (and it can't grow faster, at least in the leading $\log \ell$ term).

\subsubsection{Disorder effects}

It is generally expected that the introduction of disorder into a spin chain reduces the entanglement of eigenstates. In fact, it is often considered that one of the manifestations of the many-body localized regime is that corresponding eigenstates satisfy an area law. In the one-dimensional setting considered here this would mean that the entanglement remains bounded for large $\ell$, i.e.\ that the $\log \ell$ growth in Theorem~\ref{thm:main} should disappear in the disorder average. For the droplet spectrum of the XXZ chain this is exactly what was shown in \cite{BW18}. However, the relatively straightforward proof of the area law given in \cite{BW18}, based on what can be considered a large deviations argument provided earlier in \cite{BW17}, does not extend beyond the droplet spectrum. In fact, working in the Ising limit, we show here that above the droplet spectrum the $\log \ell$ growth will persist, even after adding disorder.

For this we will assume that the field is given by non-negative i.i.d.\ random variables $V_j$ with common distribution $\mu$ and denote disorder averaging by $\E(\cdot)$.

\begin{thm} \label{dislowK}
Let the distribution $\mu$ of the i.i.d.\ random variables $V_j$ satisfies $0 \in \supp \mu \subseteq [0,\infty)$ and $\int e^{tx}\,d\mu(x)<\infty$ for some $t>0$. Let $K\in \N$ and $\delta_0>0$. Then
\begin{equation} \label{eq:dislowK}
\liminf_{\ell\to\infty} \liminf_{L\to\infty} \frac{1}{\log \ell} \E \left( \sup_{\psi}  \mathcal{E}(\rho_\psi) \right) \ge K-1,
\end{equation}
with supremum over $\psi \in R(\chi_{[0,K+\delta_0]}(H_\Lambda^\infty(V)))$, $\|\psi\|=1$.
\end{thm}

Thus for $K\ge 2$ and any $\delta_0>0$, meaning at an arbitrarily small energy margin above the droplet spectrum, the $\log \ell$ term persists in the lower bound on the maximal possible entanglement.

Some more context on the meaning of Theorem~\ref{dislowK} is provided by the following remarks:

\begin{enumerate}
\item  We believe that \eqref{eq:dislowK} is optimal, i.e., that there also is a corresponding upper bound with constant $K-1$. But we currently have proofs of this only for $K=1$ (the area law of Beaud-Warzel) and for $K=2$. We will explain in Section~\ref{conjecture} what is lacking in our proof (and hereby invite input of readers to fill this gap, a result which should only require elementary probability). This would lead to the very neat result  that disorder lowers the pre-factor in the $\log \ell$ correction, from $2K-1$ to $K-1$, without leading to a full area law for $K\ge 2$.

\item One can argue that we are trying to do too much by considering {\it all} states with energy below $K+\delta_0$ in Theorem~\ref{dislowK}. It may still be true that only strict eigenstates satisfy an area law, not their linear combinations. This is indeed true for the Ising limit: If the $V_j$ have absolutely continuous distribution, then an argument in Appendix A of \cite{ARS} (presented there in the context of the XY chain, but applicable to more general models) shows that the spectrum of $H_\Lambda(V)$ is almost surely simple. For the Ising limit $H_\Lambda^{\infty}(V)$ this means that {\it all} its eigenfunctions are given by the product states $\phi_X$ and thus have vanishing entanglement, trivially satisfying an area law. Showing this for finite $\Delta$ remains a challenging open problem. What this demonstrates is that asking if an area law for eigenstates is a useful characteristic of MBL becomes interesting (in the sense of hard) in the XXZ chain only if one considers energies {\it above} the droplet spectrum.

\item Rigorous lower bounds on entanglement are rare, in particular for systems with disorder (\cite{Beaudetal} has a $\log$-corrected lower bound for {\it deterministic} harmonic oscillator systems). But we mention the recent \cite{Mulleretal} which shows a logarithmic lower bound for the entanglement of a disordered free Fermion system if the effective Hamiltonian is the random dimer variant of the Anderson model and the Fermi energy is chosen to be one of the critical energies, where the localization length of the random dimer model diverges.
    \end{enumerate}

\subsection{Outline of contents} The remaining sections of this work are structured as follows.

We heavily rely on the fact that the XXZ chain is particle number (down-spin number) conserving. In Section~\ref{sec:hardcore} we review the arising $N$-particle restrictions, sometimes referred to as the hard core particle formulation of the XXZ model. We do this in the setting of XXZ systems on general graphs, mostly for two reasons. First, for future work we want to have the results of Section~\ref{sec:hardcore} as well as the following Sections~\ref{sec:CT} and \ref{sec:projections} available in this more general form. Second, we feel that the way in which droplet boundary conditions (for the restriction of XXZ systems to subgraphs) arise in this general setting in the form of \eqref{NMagXXZ} is somewhat interesting by itself and, more generally, find the graph theoretic setting quite natural.

Section~\ref{sec:CT} establishes a Combes-Thomas bound which is a variant of the one proven in \cite{EKS}. In Section~\ref{sec:abstractCT} we first phrase this result in a natural setting of relatively form-bounded perturbations (Proposition~\ref{prop:CT}) before applying it to XXZ systems in Section~\ref{sec:CTXXZ}. We mention here that Proposition~\ref{prop:CT} is also applicable to XXZ models for higher spins, as recently considered in \cite{Fischbacher} (see Lemma~2.9 there), leading to a path on how the current work can be extended to the case of higher spins.

The Combes-Thomas bound establishes an exponential decay bound for the resolvent, which in Section~\ref{sec:projections} is turned in a corresponding bound on spectral projections, see Lemma~\ref{lemma:decaybounds}. Essentially, this is done via the standard Riesz contour integration formula \eqref{Riesz}, with the problem being that the contour has to cut through the spectrum at arbitrarily close distance to the nearest eigenvalue. But we can extend an argument in \cite{BW18} which has already shown how to deal with this.

In Section~\ref{sec:tracebound} we finally obtain a bound on $\Tr (\rho_{\Lambda_0,\psi})^\alpha$ for $0<\alpha<1$ in terms of an exponential sum, see Proposition~\ref{prop:summarybound}. That this is useful for proving our main result Theorem~\ref{thm:main} is due to the fact that the corresponding R\'{e}nyi entropies provide upper bounds for the von Neumann entropy.

Finding a bound on this exponential sum which is good enough to imply Theorem~\ref{thm:main} makes Section~\ref{sec:proofmain} the core technical part of our work, see Proposition~\ref{prop:summarybound}. A key ingredient to the calculations done here is to have a good understanding of how to describe the closest $K$-cluster configuration to any given $N$-particle configuration $\{x_1<x_2<\ldots<x_N\} \subseteq \Z$. Establishing sufficient understanding of this is the content of Lemmas~\ref{closestdrop} and \ref{lem:closestcluster} in Appendix~\ref{app:aux}. The resulting Theorem~\ref{thm:K}, in particular the bound \eqref{eq:alphaRenyi}, is actually the strongest version of our result from which Theorem~\ref{thm:main} follows by taking the appropriate limits. A curious aspect of our proof is that at the very end we let $\alpha\to 0$, while, due to monotonicity properties of the R\'{e}nyi entropies, the values $\alpha$ close to $1$ should give the better bounds. What happens here is that in formulating the entanglement bound in Theorem~\ref{thm:main} in the form \eqref{eq:main} all the detrimental effects of values $\alpha \sim 0$ become lower order and disappear in the limit $\ell\to\infty$ (this is a bit as saying that the sum of the sequence $e^{-\alpha n}$, $n=0,1,2,\ldots$, is of the same order of magnitude as its first term $1$, for any fixed $\alpha>0$).

The final Section~\ref{sec:Ising} contains our discussion of the Ising limit and, in particular, the proofs of Theorems~\ref{thm:detasymp} and \ref{dislowK} (as well as the description of a conjecture which we had to leave open in Section~\ref{conjecture}). The arguments here are all completely elementary. As the bounds found in this limiting case can be seen as providing some guidance for what the expect more generally in the Ising phase of the XXZ chain, this section is written so that it can be read independently of the rest of the paper.

\vspace{.3cm}

\noindent {\bf Acknowledgements:} C.\ F.\ and G.\ S.\ are grateful to the Insitut Mittag-Leffler in Djursholm, Sweden, where some of this work was done as part of the program Spectral Methods in Mathematical Physics in Spring 2019. We would also like to acknowledge useful discussions with A.\ Klein and B.\ Nachtergaele.

\section{Hard core particle formulation of the XXZ Hamiltonian on general graphs} \label{sec:hardcore}

A key property of the operators $H_\Lambda(V)$ is conservation of the number of down-spins (or particles/magnons), i.e., that it is the direct sum of its $N$-particle restrictions, $N=0,1,\ldots,L$, compare \cite{FS14}.  We will discuss basic properties of the $N$-particle operators in the more general setting where $[1,L] \subseteq \Z$ (with next-neighbor edges) is replaced by induced finite subgraphs $\mathcal{G}' = (\mathcal{V}', \mathscr{E}')$ of a more general class of graphs $\mathcal{G} = (\mathcal{V}, \mathscr{E})$.

Specifically, let $\mathcal{G}=(\mathcal{V},\mathscr{E})$ be an undirected connected graph of bounded maximal degree $d_{\max}$ with countable vertex set $\mathcal{V}$ and edge set $\mathscr{E} \subseteq \{\{x,y\}: x, y \in \mathcal{V}, x \not=y\}$. We will write $x\sim y$ for $\{x,y\} \in \mathscr{E}$.

For any $\mathcal{V}'\subseteq\mathcal{V}$, we construct the \emph{induced} subgraph $\mathcal{G}'=(\mathcal{V}',\mathscr{E}')$ by defining
\begin{equation}
\mathscr{E}':=\{\{x,y\}\subseteq\mathscr{E}: x,y\in\mathcal{V}'\}\:.
\end{equation}

In our later applications $\mathcal{G}'$ will be a finite subgraph of $\mathcal{G}$, but at this stage this is not necessary. In particular, we allow the case $\mathcal{G}' = \mathcal{G}$.

\begin{itemize}

\item (A1) We will always assume that the subgraph $\mathcal{G}'$ is {\it geodesic} in the sense that $d'(x,y)=d(x,y)$ for all $x, y \in \mathcal{V}'$.

\end{itemize}

Here $d(\cdot,\cdot)$ and $d'(\cdot,\cdot)$ denote the graph distances on $\mathcal{G}$ and $\mathcal{G}'$, respectively. In particular, this means that $\mathcal{G}'$ is connected.

The boundary  of $\mathcal{G}'$ relative to $\mathcal{G}$ is given by
\begin{equation}
\partial_{\mathcal G} \mathcal{G}' := \{x\in \mathcal{V}': \exists\ y\in \mathcal{V} \setminus \mathcal{V}' \text{ such that }x \sim y \}.
\end{equation}
For $x\in \partial_{\mathcal G} \mathcal{G}'$ we set
\begin{equation}
n(x) := |\{y\in \mathcal{V} \setminus \mathcal{V}': x \sim y\}|.
\end{equation}
For $1<\Delta\le \infty$, the XXZ Hamiltonian on $\mathcal{G}$ restricted to $\mathcal{G}'$ with droplet boundary condition is formally given by
\begin{equation}
H_{{\mathcal G}'} = \sum_{x,y \in \mathcal{V}', x\sim y} h_{x,y} + \frac{1}{2}\left(1-\frac{1}{\Delta}\right) \sum_{x \in \partial_{\mathcal{G}} \mathcal{G}'} n(x) \mathcal{N}_x
\end{equation}
with $h_{x,y}$ and $\mathcal{N}_x$ as in Section~\ref{sec:results}.
Adding a non-negative field $V:\mathcal{V}' \to [0,\infty)$ refers to the formal Hamiltonian
\begin{equation} \label{XXZformal}
H_{{\mathcal G}'}(V) = H_{{\mathcal G}'} + \sum_{x\in \mathcal{V}'} V(x) \mathcal{N}_x.
\end{equation}
If $\mathcal{V}'$ is finite, it can be readily verified that $H_{{\mathcal G}'}(V)$ is self-adjoint on the finite-dimensional tensor product $\mathcal{H}_{\mathcal{V}'} = \bigotimes_{x\in \mathcal{V}'} \C^2$. But $H_{{\mathcal G}'}(V)$ can also be given in rigorous sense if $\mathcal{V}'$ is infinite. This is best seen through the equivalent hard-core particle formulation of the XXZ Hamiltonian, which will be introduced in the remainder of this section.

\subsection{Symmetric products of graphs and subgraphs}

We recall the following
\begin{defn}[{\cite[Def.\ II.3]{FS18}}] For any $N\in\N$ such that $N\leq |\mathcal{V}|$, let the {\emph{$N$-th symmetric product}} $\mathcal{G}_N=(\mathcal{V}_N,\mathscr{E}_N)$ of $\mathcal{G}$ be the graph with vertex set
\begin{equation}
\mathcal{V}_N=\{X\subseteq\mathcal{V}:|X|=N\}
\end{equation}
and edge set
\begin{equation}
\mathscr{E}_N=\{\{X,Y\}:X,Y\in\mathcal{V}_N, X\triangle Y\in \mathscr{E}\}\:.
\end{equation}
Here $X\triangle Y=(X\setminus Y)\cup (Y\setminus X)$ is the symmetric difference of the sets $X$ and $Y$. Moreover, let $d_N(X,Y)$ denote the graph distance between two vertices $X,Y\in\mathcal{V}_N$ on $\mathcal{G}_N$.
\end{defn}
Note that (see e.g.\ \cite[Remark II.6]{FS18}) for any $X,Y\in\mathcal{V}_N$ with labeled elements $X=\{x_1,\dots,x_N\}$ and $Y=\{y_1,\dots,y_N\}$, we have the following expression for the graph distance $d_N(X,Y)$ in terms of the distance $d(\cdot,\cdot)$ on the original graph:
\begin{equation} \label{distformula}
d_N(X,Y)=\min_{\pi\in\mathfrak{S}_N}\sum_{j=1}^N d(x_j,y_{\pi(j)})\:,
\end{equation}
where $\mathfrak{S}_N$ denotes the group of permutations of $\{1,2,\dots,N\}$.

For a subgraph $\mathcal{G}'=(\mathcal{V}',\mathscr{E}')$, we denote by $\mathcal{G}'_N = (\mathcal{V}'_N, \mathscr{E}'_N)$, with $N \in \N$ such that $N \le |\mathcal{V}'|$, the $N$-th symmetric product of the $\mathcal{G}'$ (where it is easily seen that $\mathcal{G}'_N$ is a subgraph of $\mathcal{G}_N$). By \eqref{distformula} and the assumption that $\mathcal{G}'$ is geodesic in $\mathcal{G}$, this implies that $\mathcal{G}'_N$ is geodesic in $\mathcal{G}_N$, i.e.,
\begin{equation}
d_N'(X,Y) = d_N(X,Y) \quad \mbox{for all $X, Y \in \mathcal{V}'_N$}.
\end{equation}

Finally, note that the degree function $D_{\mathcal{G}}^N$ of the graph $\mathcal{G}_N$ is given by
\begin{equation} \label{GNdegree}
D_{\mathcal{G}}^N(X) = |\{Y \in \mathcal{V}_N: \{X,Y\} \in \mathscr{E}_N\}| = |\partial X|,
\end{equation}
the surface measure of $X$ in $\mathcal{G}$, i.e., the cardinality of $\partial X=\{\{x,y\}\in\mathscr{E}: x\in X, y\notin X\}$.

\subsection{The XXZ Hamiltonian on general graphs}

For any graph $\mathcal{G}$, subgraph $\mathcal{G}'$, and $N\in \N$ as above, consider the Hilbert space $\mathcal{H}_{\mathcal{G}'}^N:=\ell^2(\mathcal{V}'_N)$ and define the $N$-magnon (or $N$-particle) sector of the XXZ Hamiltonian, restricted to $\mathcal{G}'$ with {\it droplet boundary conditions}, as the operator $H^N_{\mathcal{G}'}$  given by
\begin{equation} \label{NMagXXZ}
H_{\mathcal{G}'}^N:=-\frac{1}{2\Delta}\mathcal{L}^N_{\mathcal{G}'}+\frac{1}{2}\left(1-\frac{1}{\Delta}\right)D^N_\mathcal{G}
\end{equation}
on the Hilbert space $\mathcal{H}_{\mathcal{G}'}^N = \ell^2(\mathcal{V}'_N)$ for $N\ge 1$. We also set $H_{\mathcal{G}'}^0 = 0$ on any one-dimensional Hilbert space $\mathcal{H}_{\mathcal{G}'}^0$. Here $\Delta>1$ is the anisotropy parameter and $\mathcal{L}^N_{\mathcal{G}'}$ is the graph Laplacian on $\mathcal{G}'_N$, which acts as
\begin{equation}
\left(\mathcal{L}^N_{\mathcal{G}'} f\right)(X)=\sum_{Y:\{X,Y\}\in\mathscr{E}'_N}(f(Y)-f(X)) \quad\mbox{for any $X\in\mathcal{V}'_N$}.
\end{equation}
Here, being slightly sloppy in order to avoid too many indices, the function $D^N_\mathcal{G}$ from \eqref{GNdegree} is used in \eqref{NMagXXZ} also to denote the multiplication operator in $\ell^2(\mathcal{V}'_N)$ by the restriction of this function to $\mathcal{V}'_N$.
 The assumption that the original graph $\mathcal{G}$ has bounded degree (by $d_{\max}$) yields that all $\mathcal{G}_N$ and $\mathcal{G}'_N$ have bounded degree (by $Nd_{\max}$), implying that the operators $D_{\mathcal{G}}^N$ and $\mathcal{L}_{\mathcal{G}'}^N$, and thus also $H_{\mathcal{G}'}^N$, are bounded and self-adjoint (including in the case of infinite $\mathcal{V}'$).

That the first term in \eqref{NMagXXZ} involves $\mathcal{G}'$ while the second term involves $\mathcal{G}$ is not a typo, but a reflection of the fact that in the restrictions of the XXZ Hamiltonian to subgraphs we use a suitable form of boundary conditions. More precisely, for $X \in \mathcal{V}'_N$ the difference of the degree $D^N_\mathcal{G}(X)$ in the full graph and the degree $D^N_{\mathcal{G}'}(X)$  in the subgraph (or the respective surface measures of $X$ within $\mathcal{G}$ and $\mathcal{G}'$),
\begin{equation}
\mathcal{B}^N_{\mathcal{G}'}(X) := D^N_\mathcal{G}(X) - D^N_{\mathcal{G}'}(X),
\end{equation}
is given by
\begin{equation}
\mathcal{B}^N_{\mathcal{G}'}(X)=|\{\{x,y\}\in\mathscr{E}: x\in X,y\notin\mathcal{V}'\}|,
\end{equation}
the number of edges (in $\mathcal{G}$)  leading from a point in $X$ to a point in $\mathcal{V}\setminus\mathcal{V}'$.

It will be convenient to also consider the \emph{adjacency operator} $A_{\mathcal{G}'}^N$ on $\mathcal{G}_N'$ defined by
\begin{equation}
(A_{\mathcal{G}'}^N f)(X):= \sum_{Y:\{X,Y\}\in\mathscr{E}'_N} f(Y) = (\mathcal{L}^N_{\mathcal{G}'}f)(X)+(D_{\mathcal{G}'}^N f)(X).
\end{equation}
We can re-interpret \eqref{NMagXXZ} as
\begin{align} \label{NMagXXZ2}
H_{\mathcal{G}'}^N:&=-\frac{1}{2\Delta}\mathcal{L}^N_{\mathcal{G}'}+\frac{1}{2}\left(1-\frac{1}{\Delta}\right)D^N_{\mathcal{G}'} + \frac{1}{2}\left(1-\frac{1}{\Delta}\right) \mathcal{B}^N_{\mathcal{G}'}\\
&=-\frac{1}{2\Delta}A_{\mathcal{G}'}^N+\frac{1}{2}D_{\mathcal{G}'}^N+\frac{1}{2}\left(1-\frac{1}{\Delta}\right)\mathcal{B}_{\mathcal{G}'}^N\:, \notag
\end{align}
thinking of $-\frac{1}{2\Delta}\mathcal{L}^N_{\mathcal{G}'}+\frac{1}{2}\left(1-\frac{1}{\Delta}\right)D^N_{\mathcal{G}'}$ as the {\it internal} $N$-magnon XXZ Hamiltonian on the subgraph (with ``free'' boundary condition) and as $\frac{1}{2}\left(1-\frac{1}{\Delta}\right) \mathcal{B}^N_{\mathcal{G}'}$ as a non-negative boundary field. This is what we have already referred to as droplet boundary conditions. The reason for this choice of terminology (which goes back to at least \cite{NS}) will become clear in the sequel.

We mention that for the choice $\mathcal{G}'=\mathcal{G}$ we obviously have $\mathcal{B}^N_{\mathcal{G}}=0$, so that
$H_{\mathcal{G}}^N=-\frac{1}{2\Delta}\mathcal{L}^N_{\mathcal{G}}+\frac{1}{2}\left(1-\frac{1}{\Delta}\right)D^N_\mathcal{G}$, i.e., the XXZ Hamiltonian on the full graph does not have a boundary field.

We finish this section by defining the full XXZ Hamiltonian $H_{\mathcal{G}'}$ on each subgraph $\mathcal{G}'$ as
\begin{equation} \label{eq:subgraphop}
H_{\mathcal{G}'}:=\bigoplus_{N=0}^{|\mathcal{V}'|}H_{\mathcal{G}'}^N \quad \mbox{on the Hilbert space}\quad \mathcal{H}_{\mathcal{G}'} := \bigoplus_{N=0}^{|\mathcal{V}'|} \mathcal{H}_{\mathcal{G}'}^N\:.
\end{equation}
This includes the XXZ Hamiltonian on the full graph via $\mathcal{G}' = \mathcal{G}$. This is a (generally unbounded) self-adjoint operator.

Given a graph $\mathcal{G}=(\mathcal{V},\mathscr{E})$, we call a non-negative function $V:\mathcal{V}\rightarrow [0,\infty)$ a \emph{background field}. For any $f\in\ell^2(\mathcal{V}_N)$, this defines a self-adjoint multiplication operator $V_{\mathcal{G}}^N$ (the $N$-body potential) on $\ell^2(\mathcal{V}_N)$ via
\begin{equation}
(V_{\mathcal{G}}^Nf)(X)=\left(\sum_{x\in X}V(x)\right)f(X), \quad X \subseteq \mathcal{V}_N.
\end{equation}
 It is obvious how to restrict the $V_{\mathcal{G}}^N$ to the $N$-particle subspaces of subgraphs $\mathcal{G}' = (\mathcal{V}', \mathscr{E}')$.

We define the $N$-magnon operator with background field $V$ on $\ell^2(\mathcal{V}'_N)$ as
\begin{equation} \label{Nmag+field}
H_{\mathcal{G}'}^N(V) = H_{\mathcal{G}'}^N + V_{\mathcal{G}'}^N.
\end{equation}
Finally, the XXZ Hamiltonian on $\mathcal{G}'$ with field $V$ in {\it hard-core particle form} becomes
\begin{equation} \label{XXZhardcore}
H_{\mathcal{G}'}(V) = \bigoplus_{N=0}^{|\mathcal{V}'|} H_{\mathcal{G}'}^N(V).
\end{equation}

\begin{remark} \label{ident}
That \eqref{XXZformal} and \eqref{XXZhardcore} are the same operators for finite $\mathcal{V}'$ is seen by identifying the up-down spin product basis vectors $\phi_X$, with down-spins at the sites $X\subseteq \mathcal{V}'$, with the canonical basis vectors in $\mathcal{H}_{\mathcal{G}'} := \bigoplus_{N=0}^{|\mathcal{V}'|} \ell^2(\mathcal{V}'_N)$ (in particular, for $N=0$, we identify $\ell^2(\mathcal{V}'_0)$ with the one-dimensional space spanned by the ``vacuum vector'' $\phi_\emptyset$). For infinite $\mathcal{V}'$, \eqref{XXZformal} is best understood as being defined through \eqref{XXZhardcore}. This says that $H_{\mathcal{G}'}(V)$ acts on $\phi_X$, $X$ finite, via the right hand side of \eqref{XXZformal}, and is a (generally unbounded) self-adjoint operator on the Hilbert space completion of the span of the $\phi_X$ (and is essentially self-adjoint on the latter).
\end{remark}

\subsection{Droplet configurations}
We define
\begin{equation}
D_{\mathcal{G}_N,\min}:=\min_{X\in\mathcal{V}_N} D_\mathcal{G}^N(X).
\end{equation}
Due to \eqref{GNdegree}, configurations $X\in \mathcal{V}_N$ for which this minimum is attained are solutions to the isoperimetric problem on the graph $\mathcal{G}$, i.e., sets of given volume $N$ with minimal surface area. We will call such configurations $N$-droplets.

In our choice of subgraphs $\mathcal{G}' = (\mathcal{V}', \mathscr{E}')$, we will from now on assume that

\begin{itemize}

\item (A2) $\mathcal{V}'$ contains at least one $N$-droplet for each $N\in \N$ with $N \le |\mathcal{V}'|$, i.e.,
\begin{equation}
\min_{X\in\mathcal{V}_N} D_\mathcal{G}^N(X) = \min_{X\in\mathcal{V}'_ N} D_\mathcal{G}^N(X).
\end{equation}
\end{itemize}
In particular, if $\mathcal{V}'$ is finite this means that $\mathcal{V}'$ itself is a droplet in $\mathcal{G}$.

\begin{ex} For $\mathcal{G}$, consider lattices $\Z^{d_1}$ or strips $\Z^{d_1}\times\{1,2,\dots,M\}^{d_2}$  for positive integers $d_1$ and $d_2$, where the edge set is given by $\ell^1$--next neighbors. Canonical examples for subgraphs $\mathcal{G}'$ satisfying assumptions (A1) and (A2), then are finite boxes of the form $[-L,L]^{d_1}$ and $[-L,L]^{d_1}\times\{1,2,\dots,M\}^{d_2}$, respectively. Here the validity of (A2) needs a little bit of thought about isoperimetric problems on graphs, similar to Appendix B.2 of \cite{FS18}.
\end{ex}

Furthermore, for any $k=1,2,\dots$, let
\begin{equation} \label{vNkprime}
\mathcal{V}'_{N,k} := \{X\in\mathcal{V}'_N: D_\mathcal{G}^N(X)<D_{\mathcal{G}_N,\min}+k\}
\end{equation}
and $\overline{\mathcal{V}'}_{N,k}:=\mathcal{V}'_N\setminus\mathcal{V}'_{N,k}$. By $P_{\mathcal{G}'_N,k}$ and $\overline{P}_{\mathcal{G}'_N,k}=\idty-P_{\mathcal{G}'_N,k}$ we denote the orthogonal projections onto $\ell^2(\mathcal{V}'_{N,k})$ and $\ell^2(\overline{\mathcal{V}'}_{N,k})$, respectively.

The following facts will be crucial:

\begin{prop} \label{prop:apbounds}
Let $\mathcal{G}$ be a countably
infinite, connected graph and $\mathcal{G}'$ any non-trivial subgraph of $\mathcal{G}$ satisfying (A1) and (A2). Then,
\begin{itemize}
\item[(i)] $H_{\mathcal{G}}\geq 0$ and $0$ is a simple, isolated eigenvalue of $H_{\mathcal{G}}$ with spectral gap satisfying
\begin{equation}
\min \sigma(H_{\mathcal{G}}) \setminus \{0\} \ge \frac{1}{2} \left(1-\frac{1}{\Delta} \right).
\end{equation}
\item[(ii)] $\overline{P}_{\mathcal{G}'_N,k}{H}_{\mathcal{G}'}^N\overline{P}_{\mathcal{G}'_N,k}\geq \frac{1}{2}\left(1-\frac{1}{\Delta}\right)(D_{\mathcal{G}_N,\min}+k)\overline{P}_{\mathcal{G}'_N,k}$ for all $N$ and $k$.
\end{itemize}
\end{prop}

\begin{proof}
As $\mathcal{G}$ is infinite and connected, every finite $X\subseteq \mathcal{V}$ has non-empty boundary $\partial X$. Thus $D_{\mathcal{G}_N,\min} \ge 1$ for all $N\in \N$ and therefore ${H}_{\mathcal{G}}^N \ge \frac{1}{2}(1-\frac{1}{\Delta})$ by (\ref{NMagXXZ}), using non-negativity of $-\frac{1}{2\Delta}\mathcal{L}^N_{\mathcal{G}'}$. This gives (i), as ${H}_{\mathcal{G}}^0=0$ on a one-dimensional space.
Property (ii) follows similarly from \eqref{NMagXXZ} and the definition \eqref{vNkprime}.
\end{proof}

\section{Combes--Thomas Bounds} \label{sec:CT}

Here we start with a Combes-Thomas bounds for discrete Schr\"odinger-type operators, in which the hopping part satisfies a form bound relative to the potential. This will then be applied to the hard-core particle operators \eqref{NMagXXZ}.

\subsection{An abstract Combes-Thomas bound} \label{sec:abstractCT}
Consider a discrete Schr\"odinger-type operator of the form
\begin{equation} \label{eq:schroedop}
H=-gA+W\:,
\end{equation}
defined on a countable (finite or infinite), non-directed and connected graph $G=(V,E)$. By $d(\cdot,\cdot)$ we denote the graph distance on $V$. Here $g>0$ is a parameter and we assume $A$ to be a weighted adjacency matrix on $G$ of the form
\begin{equation} \label{eq:finiterange}
(A\psi)(x)=\sum_{y:d(x,y)\leq s_{\max}}A(x,y)\psi(y)
\end{equation}
for some $s_{\max} \in \N$, with $A(x,y)=A(y,x) \ge 0$.
Moreover, $W$ is assumed to be a strictly positive multiplication operator (hence in particular boundedly invertible). For later convenience, we define $W_0:=\inf_{x\in V}W(x)>0$.

For our intended applications it will suffice to consider cases where $A$ and $W$ are bounded operators (this holds for $A$, for example, if $G$ has bounded degree and $A(x,y)$ is bounded). But note that \eqref{eq:schroedop} could also be defined as a self-adjoint form sum via the KLMN theorem if $gc<1$ with $c$ from the relative form bound \eqref{eq:relbound} below.

\begin{prop} \label{prop:CT}
	Consider $H$ as in Equation \eqref{eq:schroedop} assuming $A$ is of the form \eqref{eq:finiterange}. Moreover, assume that there exists $c>0$ such that
	\begin{equation}\label{eq:relbound}
	-cW\leq A\leq cW\:.
	\end{equation}
	
	Lastly, let $z\notin\sigma(H)$ such that there exists $\kappa_z>0$ for which
	\begin{equation} \label{initbound}
	\left\|W^{1/2}(H-z)^{-1}W^{1/2}\right\|\leq \frac{1}{\kappa_z}<\infty\:.
	\end{equation}
	Then for all subsets $\mathcal{A,B}\subseteq V$, we have
	\begin{equation}
	\left\|\chi_\mathcal{A}\left(H-z\right)^{-1}\chi_\mathcal{B}\right\| \leq \frac{1}{W_0} \left\|\chi_\mathcal{A}W^{1/2}\left(H-z\right)^{-1}W^{1/2}\chi_\mathcal{B}\right\|\leq \frac{2}{W_0\kappa_z}\,e^{-\eta_z d(\mathcal{A,B})}\:,
	\end{equation}
	where
	\begin{equation} \label{etaz}
	\eta_z=\frac{1}{s_{\max}} \log\left(1+\frac{\kappa_z}{2gc}\right).
	\end{equation}
	\label{prop:ct1}
\end{prop}
Here and in the following $\chi_\mathcal{A}$ denotes the orthogonal projection on the configuration space $\mathcal{A}\subseteq V$.

\begin{proof}
	Up to some abstractions and modifications, we follow the argument in the proof of Proposition~4.1 in \cite{EKS}.
	Firstly, observe that \eqref{eq:relbound} implies
	\begin{equation}\label{eq:relbound2}
	-c \leq W^{-1/2}AW^{-1/2}\leq c\:.
	\end{equation}
	Now, for any $\mathcal{A}\subseteq V$, let $\rho_\mathcal{A}$ be the operator of multiplication by $d(\mathcal{A},\cdot)$, i.e., $(\rho_\mathcal{A}\psi)(x):=d(\mathcal{A},x)\psi(x)$. For any $\eta>0$ define
	\begin{equation}
	H_\eta:=e^{-\eta\rho_\mathcal{A}}He^{\eta\rho_\mathcal{A}}
	\end{equation}	
	and $K_\eta:=H_\eta-H$. Observe that
	\begin{equation}
	K_\eta=-g\left(e^{-\eta\rho_\mathcal{A}}Ae^{\eta\rho_\mathcal{A}}-A\right)\:.
	\end{equation}
	Now, for any $\psi\in \ell^2(V)$, consider
	\begin{eqnarray}
	\lefteqn{\left\|W^{-1/2}K_\eta W^{-1/2}\psi\right\|^2} \\
	&=&g^2 \sum_{x} \left| \sum_{y: d(x,y)\leq s_{\max}}W^{-1/2}(x)W^{-1/2}(y)\left(e^{\eta(\rho_{\mathcal{A}}(y)-\rho_{\mathcal{A}}(x))}-1\right) A(x,y) \psi(y) \right|^2 \notag \\
	& \leq& g^2 \left(e^{\eta s_{\max}}-1\right)^2 \sum_{x} \left( \sum_{y: d(x,y)\leq s_{\max}}W^{-1/2}(x)W^{-1/2}(y) A(x,y) |\psi(y)| \right)^2 \notag \\
	& =&g^2\left(e^{\eta s_{\max}}-1\right)^2 \left\|W^{-1/2}AW^{-1/2}|\psi|\right\|^2\overset{\eqref{eq:relbound2}}{\leq}g^2 c^2\left(e^{\eta s_{\max}}-1\right)^2 \|\psi\|^2\:, \notag
	\end{eqnarray}	
	which implies
	\begin{equation} \label{eq:dildiffbound}
	\left\|W^{-1/2}K_{\eta}W^{-1/2}\right\|\leq cg(e^{\eta s_{\max}}-1)\:.
	\end{equation}
	For $\eta = \eta_z$ as in \eqref{etaz} it follows that
	\begin{eqnarray} \label{halfbound}
	\|W^{-1/2} K_\eta (H-z)^{-1} W^{1/2} \| & = & \| W^{-1/2} K_\eta W^{-1/2} W^{1/2} (H-z)^{-1} W^{1/2} \| \\
	& \le & \frac{cg (e^{\eta s_{\max}} -1)}{\kappa_z} = \frac{1}{2}. \notag
	\end{eqnarray}
	Using the resolvent identity
	\begin{equation}
	W^{1/2} (H_\eta -z)^{-1} W^{1/2} (\idty + W^{-1/2} K_\eta (H-z)^{-1} W^{1/2}) = W^{1/2} (H-z)^{-1} W^{1/2}
	\end{equation}
	and that $\|(\idty+A)^{-1}\| \le (1-\|A\|)^{-1}$ for $\|A\|<1$ we conclude from \eqref{initbound} and \eqref{halfbound} that
	\begin{eqnarray}
	\lefteqn{\|W^{1/2} (H_\eta-z)^{-1} W^{1/2}\|} \\
	& \le &  \|W^{1/2} (H-z)^{-1} W^{1/2}\| \|(\idty+ W^{-1/2} K_\eta (H-z)^{-1} W^{1/2})^{-1}\| \le \frac{2}{\kappa_z}. \notag
	\end{eqnarray}
	From this, we get
	\begin{eqnarray}
	\left\|\chi_\mathcal{A}W^{1/2}(H-z)^{-1}W^{1/2}\chi_\mathcal{B}\right\|
	 &=&  \left\|\chi_\mathcal{A}e^{\eta\rho_\mathcal{A}}W^{1/2}(H_\eta-z)^{-1}W^{1/2}e^{-\eta\rho_\mathcal{A}}\chi_\mathcal{B}\right\|  \\
	& \leq& \left\|W^{1/2}(H_\eta-z)^{-1}W^{1/2}\right\|\left\|e^{-\eta\rho_\mathcal{A}}\chi_\mathcal{B}\right\| \notag \\
	& \leq & \frac{2}{\kappa_z}e^{-\eta d(\mathcal{A,B})}\:, \notag
	\end{eqnarray}
	which is the desired result.
\end{proof}
\begin{remark} Note that if $W$ is strictly positive and $A$ is bounded, we could always choose $c=\frac{\|A\|}{W_0}$ in \eqref{eq:relbound}, since we have the estimate
	\begin{equation}
	-\frac{\|A\|}{W_0}W\leq A\leq \frac{\|A\|}{W_0}W\:.
	\end{equation}
	However, in later applications of this proposition (cf.\ Corollary \ref{cor:ctxxz}, Equation \eqref{eq:XXZformbound}), we will consider {\it families} of operators $H_N=-gA_N+W_N$ with $N\in\N$, for which $\|A_N\|/W_{N,0}$ is not uniformly bounded in $N$. Nevertheless, for this specific case, we will show the existence of a $c>0$ -- independent of $N$ -- such that $-cW_N\leq A_N \leq cW_N$ for each $N\in\N$.
\end{remark}
Now, for any $K\in\R^+$, we decompose the vertex set $V$ into the two disjoint sets $V_K:=\{x\in V: W(x)\leq K\}$ and $\overline{V}_K:=V\setminus V_K$. 
For any $\delta'>0$, assume
\begin{equation} \label{eq:apriori}
g<1/c \quad \mbox{and}\quad W_{0}<(K-\delta')\:.
\end{equation}

We then get the following
\begin{prop}
	Fix $\delta'>0$ and assume that $g,K>0$ are such that both conditions in \eqref{eq:apriori} are satisfied. Then, for any $\epsilon\in\R$, any $E\leq (1-cg)(K-\delta')$ and any $\mathcal{A,B}\subseteq \overline{V}_K$, we get
	\begin{align}
	\|\chi_\mathcal{A}(\chi_{\overline{V}_K}(H-E\pm i\epsilon)\chi_{\overline{V}_K})^{-1}\chi_\mathcal{B}\|&\leq \frac{1}{K}\|\chi_\mathcal{A}W^{1/2}(\chi_{\overline{V}_K}(H-E\pm i\epsilon)\chi_{\overline{V}_K})^{-1}W^{1/2}\chi_\mathcal{B}\|\\
	&\leq Ce^{-\eta d(\mathcal{A,B})}\:, \notag
	\end{align}
	where $C=\frac{4}{\delta'(1-cg)}, \eta=\frac{1}{s_{\max}}\log\left(1+\frac{\delta'(1-cg)}{4Kcg}\right)$ and the inverse $(\chi_{\overline{V}_K}(H-E\pm i\epsilon)\chi_{\overline{V}_K})^{-1}$ has to be understood as taken on $\ell^2(\overline{V}_K)$.
	\label{prop:ct2}
\end{prop}
\begin{proof} Note that by \eqref{eq:relbound}, we have
	\begin{equation}
	W^{-1/2}(H-E)W^{-1/2}=-gW^{-1/2}AW^{-1/2}+\idty-EW^{-1}\geq (1-cg)\idty-EW^{-1}
	\end{equation}
	and thus in particular
	\begin{align}
	W^{-1/2}\chi_{\overline{V}_K}(H-E)\chi_{\overline{V}_K}W^{-1/2}&\geq (1-cg)\chi_{\overline{V}_K}-EW^{-1}\chi_{\overline{V}_K}\\&\geq (1-cg)\chi_{\overline{V}_K}-\frac{(1-cg)(K-\delta')}{K}\chi_{\overline{V}_K}=\frac{\delta'}{K}(1-cg)\chi_{\overline{V}_K}\:, \notag
	\end{align}
	where for the second estimate, we have used that $E\leq (1-cg)(K-\delta')
	$ and $K\chi_{\overline{V}_K}\leq W\chi_{\overline{V}_K}$. This implies that, as an operator on $\ell^2(\overline{V}_K)$, $W^{-1/2}\chi_{\overline{V}_K}(H-E)\chi_{\overline{V}_K}W^{-1/2}$ is boundedly invertible with
	\begin{equation}
	\left\|W^{1/2}(\chi_{\overline{V}_K}(H-E)\chi_{\overline{V}_K})^{-1}W^{1/2}\right\|\leq \frac{K}{\delta'(1-cg)}\:.
	\end{equation}
	A slight modification of this argument -- see e.g. \cite[Lemma 4.3]{EKS} -- then yields
	\begin{equation}
	\left\|W^{1/2}(\chi_{\overline{V}_K}(H-E\pm i\epsilon)\chi_{\overline{V}_K})^{-1}W^{1/2}\right\|\leq \frac{2K}{\delta'(1-cg)}
	\end{equation}
	for any $\epsilon\in\R$. We now finish the proof by applying Proposition \ref{prop:ct1} to the operator
	\begin{equation}
	\chi_{\overline{V}_K}H\chi_{\overline{V}_K}=\chi_{\overline{V}_K}(-gA+W)\chi_{\overline{V}_K}
	\end{equation}
	defined on $\ell^2(\overline{V}_K)$. Clearly, $\overline{A}_K:=\chi_{\overline{V}_K}A\chi_{\overline{V}_K}$ is still a weighted adjacency matrix, since for any $\psi\in\ell^2(\overline{V}_K)$, we get
	\begin{equation}
	(\overline{A}_K\psi)(x)=\sum_{y\in\overline{V}_K:d(x,y)\leq s_{\max}}A(x,y)\psi(y)
	\end{equation}
	and we can still estimate
	\begin{equation}
	-cW\chi_{\overline{V}_K}\leq \overline{A}_K\leq cW\chi_{\overline{V}_K}\:.
	\end{equation}
	For any $\epsilon\in\R$, any $\mathcal{A,B}\subseteq \overline{V}_K$ and any $E\leq (1-cg)(K-\delta')$, we therefore get
	\begin{align}
	\|\chi_\mathcal{A}(\chi_{\overline{V}_K}(H-E\pm i\epsilon)\chi_{\overline{V}_K})^{-1}\chi_\mathcal{B}\|&\leq \frac{1}{K}\|\chi_\mathcal{A}W^{1/2}(\chi_{\overline{V}_K}(H-E\pm i\epsilon)\chi_{\overline{V}_K})^{-1}W^{1/2}\chi_\mathcal{B}\|\\
	&\leq \frac{4}{\delta'(1-cg)}\left(1+\frac{\delta'(1-cg)}{4Kcg}\right)^{-{d(\mathcal{A,B})}/{s_{\max}}}\:, \notag
	\end{align}
	which is the desired result.
\end{proof}
\begin{remark} Without any modification or change in the estimate's constants, this argument can be extended to operators $H_Y:=H+Y$, where $Y$ is a multiplication operator by an arbitrary non--negative function.
\end{remark}

\subsection{Application to the XXZ Hamiltonian} \label{sec:CTXXZ}

Now, we will apply this to the $N$-particle operators $H_{\mathcal{G}'}^N$ defined in Section~\ref{sec:hardcore}, which will lead to a combination of previous results from \cite{EKS} and \cite{FS18}. We will treat the case $\mathcal{G}'=\mathcal{G}$, as the introduction of the positive boundary field in \eqref{NMagXXZ2} does not effect the proof, see Remark~\ref{rem:bf} below.

For its statement, let us introduce the following notation: For any $\mathcal{A,B}\subseteq \mathcal{V}_N$, we define their distance

\begin{equation}
d_N(\mathcal{A,B}):=\min_{X\in\mathcal{A},Y\in\mathcal{B}}d_N({X,Y})
\end{equation}
and for any graph $\mathcal{G}=(\mathcal{V},\mathscr{E})$, any $N\in \N$ with $N\leq |\mathcal{V}|$ and $k=1,2,3,\dots$ we define
\begin{equation} \label{ENk}
E_{N,k}:=\frac{1}{2}\left(1-\frac{1}{\Delta}\right)(D_{\mathcal{G}_N,\min}+k)\:.
\end{equation}

\begin{cor}[{Cf.\ also \cite[Remark\ V.3]{FS18}}] \label{cor:ctxxz}
Let $\mathcal{G} = (\mathcal{V}, \mathscr{E})$ be as in Section~\ref{sec:hardcore}.
Moreover, let $V$ be an arbitrary non-negative background field on $\mathcal{V}$ and fix $\delta>0$.
	Then, for any $N\in \N$ with $N\leq |\mathcal{V}|$, any $\mathcal{A,B}\subseteq\overline{\mathcal{V}'}_{N,k}$, any $E\in\R$ satisfying
	\begin{equation}
	E\leq E_{N,k}-\delta
	\end{equation}
	and any $\epsilon\geq 0$, the following estimate holds:
	\begin{equation} \label{eq:ct}
	\|\chi_\mathcal{A}(\overline{P}_{\mathcal{G}_N,k}(H_\mathcal{G}^N+V_\mathcal{G}^N-(E\pm i\epsilon))\overline{P}_{\mathcal{G}_N,k})^{-1}\chi_\mathcal{B}\|\leq C e^{-\mu d_N(\mathcal{A,B})}\:,
	\end{equation}
	where
	\begin{equation} \label{eq:constants}
	C=\frac{4}{\delta}\quad\mbox{and}\quad \mu=\log\left(1+\frac{\delta\Delta}{2(D_{\mathcal{G}_N,\min}+k)}\right)\:.
	\end{equation}
\end{cor}

\begin{proof} Consider the XXZ Hamiltonian $H_\mathcal{G}^N$ on $\ell^2(\mathcal{V}_N)$ given by
\begin{equation}
H_\mathcal{G}^N=-\frac{1}{2\Delta}A_\mathcal{G}^N+\frac{1}{2}D_\mathcal{G}^N\:.
\end{equation}
Firstly, note that since
\begin{equation}
\langle f,(A_\mathcal{G}^N\pm D_\mathcal{G}^N)f\rangle=\sum_{\{X,Y\}\in\mathscr{E}_N}\left|f(Y)\pm f(X)\right|^2\geq 0\:,
\end{equation}
we have the relative bound
\begin{equation} \label{eq:XXZformbound}
-D_\mathcal{G}^N\leq A_\mathcal{G}^N\leq D_\mathcal{G}^N\:,
\end{equation}
which is independent of $N$.
The corollary now follows from an application of Proposition \ref{prop:ct2} to $H_\mathcal{G}^N$, where  $g=\frac{1}{2\Delta}$, $c=2$ (following from $W=D_\mathcal{G}^N/2$ and \eqref{eq:XXZformbound}), $\delta'=\frac{\delta}{1-1/\Delta}$, $s_{\max}=1$ and $K=\frac{1}{2}(D_{\mathcal{G}_N,\min}+k)$.
\end{proof}
\begin{remark} \label{rem:bf}
This result is also true when considering $H_{\mathcal{G}'}^N$, where $\mathcal{G}'$ is assumed to satisfy Assumptions (A1) and (A2). Again, one can apply Proposition \ref{prop:ct2} to
$\left(-\frac{1}{2\Delta}A_{\mathcal{G}'}+\frac{1}{2}D_{\mathcal{G}'}^N\right)$
 with the same choice of parameters, while treating the boundary field term $\frac{1}{2}\left(1-\frac{1}{\Delta}\right)\mathcal{B}_{\mathcal{G}'}^N$ as a non--negative potential.
\end{remark}	

\begin{remark} Note that the dependence on $N$ and on the specific graph $\mathcal{G}$ enters the constants $C$ and $\mu$ in the previous proposition's estimate only through their dependence on $D_{N,\min}$, the solution of the discrete isoperimetric problem
	\begin{equation}
	D_{\mathcal{G}_N,\min}=\min\{|\partial X|: X\subseteq\mathcal{V}, |X|=N\}\:.
	\end{equation}
\end{remark}

\section{Decay bounds on spectral projections} \label{sec:projections}

We now extend an argument of \cite{BW18} which shows how the Combes-Thomas bound \eqref{cor:ctxxz} leads to decay bounds on spectral projections.

\begin{lemma} \label{lemma:decaybounds}
Let $\mathcal{G} = (\mathcal{V}, \mathscr{E})$ be as in Section~\ref{sec:hardcore}.  Moreover, let $V$ be an arbitrary non-negative background field on $\mathcal{V}$ and fix $\delta>0$ and $k\in\{1,2,\dots\}$. Then, for any finite subgraph $\mathcal{G}'$ of $\mathcal{G}$ satisfying assumptions (A1) and (A2), for any $N\leq |\mathcal{V}'|$ and any $\mathcal{A}\subseteq\overline{\mathcal{V}'}_{N,k}$, there exist constants $C_2, \mu_2>0$, such that the following estimate holds
\begin{equation}
\|\chi_{\mathcal{A}}\chi_{[0,E_{N,k}-\delta]}\big(
H_{\mathcal{G}'}^N(V)\big)\|\leq C_2e^{-\mu_2 d_N(\mathcal{A},\mathcal{V}'_{N,k})}\:,
\end{equation}
where $H_{\mathcal{G}'}^N(V)=H_{\mathcal{G}'}^N+V_{\mathcal{G}'}^N$.
The constants are given by:
\begin{align}
C_2=\frac{3\sqrt{5}}{2}\frac{(D_{\mathcal{G}_N,\min}+k)^{3/2}}{\min\{1,\delta^{3/2}\}}\quad\mbox{and}\quad \mu_2=\frac{\mu}{2}=\frac{1}{2}\log\left(1+\frac{\delta\Delta}{2(D_{\mathcal{G}_{N},\min}+k)}\right)\:.
\end{align}
\end{lemma}

\begin{proof} We closely follow arguments from \cite[Lemma III.1]{BW18} here. For convenience, we omit the dependence of ${H}_{\mathcal{G}'}^N(V)$ on $V$ and $N$ and just write $ H :={H}^N_{\mathcal{G}'}(V)$ instead.  Also we abbreviate $E' := E_{N,k}-\delta$, and we use $\mathcal{Q}$ to denote the spectral projection of $H_{\mathcal{G}'}^N(V)$ associated with $[0, E_{N,k}-\delta]$, i.e.,
\begin{equation}
\mathcal{Q}:=\chi_{[0,E_{N,k}-\delta]}\big(
H_{\mathcal{G}'}^N(V)\big)=\chi_{[0,E']}\big(
 H \big).
\end{equation}
Using Riesz' theorem on spectral projections, we then may write
\begin{equation} \label{Riesz}
\chi_\mathcal{A}\mathcal{Q} \chi_\mathcal{A}=-\frac{1}{2\pi i}\oint_\Gamma \chi_\mathcal{A}( H -z)^{-1}\chi_\mathcal{A}dz\:,
\end{equation}
where $\Gamma$ is the positively oriented rectangle in the complex plane with corners at $(-1\pm i)$ and $(E'\pm i)$. Note that since $\mathcal{G}'$ is a finite graph, $\sigma( H )$ consists entirely of isolated eigenvalues. Hence, if $E'$ happens to be an eigenvalue of $ H $, there always exists an $0<\varepsilon<\frac{\delta}{2}$ such that the spectral projection associated with the closed interval $[0,E'+\varepsilon]$ is equal to the previously introduced spectral projection $\mathcal{Q}$ and instead of $\Gamma$, we could choose the positively oriented rectangle with corners at $-1\pm i$ and $E'+\varepsilon \pm i$. In what follows, we will only consider the case that $E'$ is not an eigenvalue of $ H $; the other case can be shown completely analogously. The only important thing is to note that $E'+\varepsilon$ is still uniformly bounded away from $E_{N,k}$.

Firstly, note that since $\mathcal{A}\subseteq\overline{\mathcal{V}'}_{N,k}$ we get
\begin{equation}
\chi_\mathcal{A}( H -z)^{-1}\chi_\mathcal{A}=\chi_\mathcal{A}\overline{P}( H -z)^{-1}\overline{P}\chi_\mathcal{A}\:,
\end{equation}
where we abbreviated $\overline{P}:=\overline{P}_{\mathcal{G}'_N,k}$ for convenience. (We also define $P:=\idty-\overline{P}$.)
By Schur decomposition, we get
\begin{equation}
\overline{P}( H -z)^{-1}\overline{P}=r(z)+r(z)\overline{P} H  P( H -z)^{-1}P H \overline{P} r(z)\:,
\end{equation}
where we have defined
\begin{equation}
r(z):=(\overline{P}( H -z)\overline{P})^{-1}\:.
\end{equation}
By virtue of Proposition \ref{prop:apbounds}, part (ii), we know that $r(z)$ is analytic inside $\Gamma$, hence its contour integral along $\Gamma$ vanishes. We are thus left with
\begin{align}
\|\chi_\mathcal{A}\mathcal{Q}\chi_\mathcal{A}\|&=\frac{1}{2\pi}\left\|\oint_{\Gamma} \chi_\mathcal{A}r(z)\overline{P}  H  P( H -z)^{-1}P  H \overline{P} r(z)\chi_\mathcal{A}dz\right\|
\end{align}
We begin by estimating
\begin{equation} \label{eq:preest}
\left\|\oint_{\Gamma_1} \chi_\mathcal{A}r(z)\overline{P} H  P( H -z)^{-1}P H \overline{P} r(z)\chi_\mathcal{A}dz\right\|\:,
\end{equation}
where $\Gamma_1$ is the (non-closed) polygonal chain connecting the complex points: $(E'+i)\rightarrow (-1+i)\rightarrow (-1-i) \rightarrow (E'-i) $.
In particular, note that for any $z\in\Gamma_1$, we have $\text{dist}(z,\sigma( H ))\geq 1$ from which we get the estimate $\|( H -z)^{-1}\|\leq 1$.

Moreover, note that since $\overline{P} H  P =-\frac{1}{2\Delta}\overline{P} A^N_{\mathcal{G}'} P$, we get
\begin{equation}
\overline{P} H   P =\chi_{W} H  P\:,
\end{equation}
where
\begin{equation}
W=\{X\in\overline{\mathcal{V}'}_{N,k}: d_N(X,\mathcal{V}'_{N,k})=1\}\:.
\end{equation}
We thus continue our estimate of \eqref{eq:preest}:
\begin{align}
\eqref{eq:preest}&\leq \ell(\Gamma_1)\|\chi_\mathcal{A}r(z)\overline{P} H  P( H -z)^{-1}P H \overline{P} r(z)\chi_\mathcal{A}\|\\
&=\ell(\Gamma_1)\|\chi_\mathcal{A}r(z)\chi_W H  P( H -z)^{-1}P H \chi_W r(z)\chi_\mathcal{A}\| \notag \\
&\leq \ell(\Gamma_1)\|\chi_\mathcal{A}r(z)\chi_W\|\|\overline{P} H  P\|\|( H -z)^{-1}\|\|P H \overline{P}\|\|\chi_W r(z)\chi_\mathcal{A}\| \notag \\
&\leq\ell(\Gamma_1)\|\chi_\mathcal{A}r(z)\chi_W\|^2\|\overline{P} H  P\|^2 \notag \\
&\leq\ell(\Gamma_1)\|\overline{P} H  P\|^2\cdot C^2e^{-2\mu d_N(\mathcal{A},W)} \notag \\
&\leq\ell(\Gamma_1)\|\overline{P} H  P\|^2\cdot C^2e^{2\mu}e^{-2\mu d_N(\mathcal{A},\mathcal{V}'_{N,k})}\:, \notag
\end{align}
where we have used the Combes-Thomas estimate \eqref{eq:ct} for $\|\chi_\mathcal{A}r(z)\chi_W\|$ in the penultimate step and the fact that $|d_N(\mathcal{A},W)-d_N(\mathcal{A},\mathcal{V}'_{N,k})|\leq 1$ for the last inequality. Also, note that the length of $\Gamma_1$ is given by $\ell(\Gamma_1)=2E'+4$. The constants $C$ and $\mu$ are explicitly given in Corollary \ref{cor:ctxxz}, Equation \eqref{eq:constants}.

Writing $z=E'+it$, we now continue by estimating
\begin{equation} \label{eq:lineint}
\left\|\int_{-1}^{1}\chi_\mathcal{A}r(E'+it)\overline{P} H  P( H -z)^{-1}P H \overline{P} r(E'+it)\chi_\mathcal{A}dt\right\|\:.
\end{equation}
Using the resolvent identity
\begin{equation} \label{eq:resolvent}
r(E'+it)=r(E')-itr(E')r(E'+it)
\end{equation}
we continue
\begin{align} \label{eq:int1}
\eqref{eq:lineint}&\leq \left\|\int_{-1}^{1}\chi_\mathcal{A}r(z)\overline{P} H  P( H -z)^{-1}P H  \overline{P} r(E')\chi_\mathcal{A}dt\right\| \\
&+\left\|\int_{-1}^{1}\chi_\mathcal{A}r(z)\overline{P} H  P it( H -z)^{-1}P H \overline{P} r(E')r(E'+it)\chi_\mathcal{A}dt\right\| \notag
\end{align}

We start with the second summand in \eqref{eq:int1}, which can easily be estimated
\begin{align} \label{eq:estimate1}
&\leq \int_{-1}^1\|\chi_\mathcal{A}r(z)\chi_W\|\|\overline{P} H  P\|^2\|it( H -(E'+it))^{-1}\|\|r(E')\|\|r(E'+it)\|\|\chi_\mathcal{A}\|dt\\
&\leq 2 C e^{\mu}e^{-\mu d_N(\mathcal{A},\mathcal{V}'_{N,k})} \|\overline{P} H  P\|^2 \delta^{-2}\:, \notag
\end{align}
where  the last factor $\delta^{-2}$ comes from estimating the terms $\|r(E')\|\leq \delta^{-1}$ and $\|r(E'+it)\|\leq \delta^{-1}$ with the help of Proposition \ref{prop:apbounds}, {(ii)}. The estimate for the term $\|\chi_\mathcal{A}r(z)\chi_W\|$ has already been obtained above. Lastly, note that we have also used $\|it( H -(E'+it))^{-1}\|\leq 1$.

Now, in order to estimate the first summand in \eqref{eq:int1}, we apply the resolvent identity \eqref{eq:resolvent} another time to get the bound
\begin{align} \label{eq:int3}
\leq&\left\|\int_{-1}^1\chi_\mathcal{A}r(E')\overline{P} H  P( H -(E'+it))^{-1}P H \overline{P} r(E')\chi_\mathcal{A}dt\right\|\\
+&\left\|\int_{-1}^1 \chi_\mathcal{A}r(E')r(E'+it)\overline{P} H  P it ( H - (E'+it))^{-1}P H  \overline{P} r(E')\chi_\mathcal{A}dt\right\| \notag
\end{align}
The second summand in \eqref{eq:int3} can be estimated analogously as the second summand in \eqref{eq:int1}, which has been done in \eqref{eq:estimate1}. Hence, for this term we get the bound
\begin{equation}
\leq 2 C e^{\mu}e^{-\mu d_N(\mathcal{A},\mathcal{V}'_{N,k})} \|\overline{P} H  P\|^2 \delta^{-2}\:.
\end{equation}
To finish the proof, we estimate the first summand in \eqref{eq:int3} by
\begin{align} \label{eq:int5}
&\leq \|\chi_\mathcal{A}r(E')\chi_W\|^2\|\overline{P} H  P\|^2\left\|\int_{-1}^1( H -(E'+it))^{-1}dt\right\|\\
&\leq C^2e^{2\mu}e^{-2\mu d_N(\mathcal{A},\mathcal{V}'_{N,k})}\|\overline{P} H  P\|^2\left\|\int_{-1}^1( H -(E'+it))^{-1}dt\right\|\:, \notag
\end{align}
where we again argued as above for the estimate of $\|\chi_\mathcal{A}r(E')\chi_W\|$. To finish, note that for any $f\in\mathcal{H}_{\mathcal{G}'}^N$ we have by the spectral theorem
\begin{equation}\label{eq:specint}
\left\|\int_{-1}^1( H -(E'+it))^{-1}dt f\right\|^2=\sum_{n}\left|\int_{-1}^1 \frac{1}{\lambda_n-(E'+it)}dt\right|^2\|\chi_{\{\lambda_n\}}( H )f\|^2\:,
\end{equation}
Now, since $E'$ is not an eigenvalue of $ H $ we get for any $n$ that $\lambda_n-E'\neq 0$ and consequently
\begin{equation}
\left|\int_{-1}^1 \frac{1}{\lambda_n-(E'+it)}dt\right|^2\leq 4\pi^2\:.
\end{equation}
Thus,
\begin{equation}
\eqref{eq:specint}\leq 4\pi^2\sum_n\|\chi_{\{\lambda_n\}}( H ) f\|^2=4\pi^2\|f\|^2\:,
\end{equation}
from which we immediately conclude
\begin{equation}
\eqref{eq:int5}\leq 2\pi C^2e^{2\mu}e^{-2\mu d_N(\mathcal{A},\mathcal{V}'_{N,k})}\|\overline{P} H  P\|^2\:.
\end{equation}

Now, as mentioned above, we have $\overline{P}  H  P =-\frac{1}{2\Delta}\overline{P} A^N_{\mathcal{G}'}P$. Letting $\mathscr{P}(X,Y)$ denote the kernel of $\overline{P} A^N_{\mathcal{G}'}P$, where
\begin{equation}
\mathscr{P}(X,Y)=\begin{cases} 1 \quad \mbox{if}\quad X\in\overline{\mathcal{V}'}_{N,k}, Y\in\mathcal{V}'_{N,k}\:\:\mbox{and}\:\: \{X,Y\}\in\mathscr{E}_N\\ 0 \quad \mbox{else,}
\end{cases}
\end{equation}
we find that for any fixed $Y\in\mathcal{V}'_{N,k}$, we can make the following estimate
\begin{equation}
\sum_{X\in\mathcal{V}'}\mathscr{P}(X,Y)=\sum_{X\in\overline{\mathcal{V}'}_{N,k}: \{X,Y\}\in\mathscr{E}_N}1\leq\sum_{X\in \mathcal{V}'_N: \{X,Y\}\in\mathscr{E}_N}1=D_{\mathcal{G}'}^N(Y)\leq D_{\mathcal{G}_{N},\min}+k-1\:.
\end{equation}
Hence, we get
\begin{equation}
\|\overline{P}  H  P\|\leq \frac{D_{\mathcal{G}_N,\min}+k-1}{2\Delta}\:.
\end{equation}
\begin{remark}
A more refined estimate for $\|\overline{P}  H  P\|$ -- depending on the specific structure of the underlying graph $\mathcal{G}$ -- can be found in \cite[Lemma 3.1]{FS18}.
\end{remark}
 By carefully collecting all the above made estimates, we find
\begin{align}
	\|\chi_\mathcal{A}\mathcal{Q}\chi_\mathcal{A}\|&\leq \|\overline{P}  H  P\|^2\left(\left(\frac{E'+2}{\pi}+1\right)C^2e^{2\mu}e^{-2\mu d_N(\mathcal{A},\mathcal{V}'_{N,k})}+\frac{2Ce^\mu}{\pi\delta^2}e^{-\mu d_N(\mathcal{A},\mathcal{V}'_{N,k})}\right)\\&\leq \|\overline{P}  H  P\|^2\left(\left(\frac{E_{N,k}+2}{\pi}+1\right)C^2e^{2\mu}+\frac{2Ce^\mu}{\pi\delta^2}\right)e^{-\mu d_N(\mathcal{A},\mathcal{V}'_{N,k})}\:. \notag
\end{align}
 Using that $\|\chi_\mathcal{A}\mathcal{Q}\chi_\mathcal{A}\|=\|\chi_\mathcal{A}\mathcal{Q}\|^2$, we therefore get
 \begin{align}
 \|\chi_\mathcal{A}\mathcal{Q}\|\leq& \|\overline{P}  H  P\|\left(\left(\frac{E_{N,k}+2}{\pi}+1\right)C^2e^{2\mu}+\frac{2Ce^\mu}{\pi\delta^2}\right)^{1/2}e^{-\frac{\mu}{2} d_N(\mathcal{A},\mathcal{V}'_{N,k})}\\
 \leq &\frac{3\sqrt{5}}{2}\frac{(D_{\mathcal{G}_N,\min}+k)^{3/2}}{\min\{1,\delta^{3/2}\}}e^{-\frac{\mu}{2}d_N(\mathcal{A},\mathcal{V}'_{N,k})}\:, \notag
 \end{align}
where the last estimate can be obtained by using the explicit expressions for $C$ and $\mu$, noting that $\Delta>1$ as well as $D_{\mathcal{G}_N,\min}+k\geq 2$.
\end{proof}

\section{Entanglement bounds for the chain} \label{sec:tracebound}

From now on, we restrict our considerations to the finite chain, i.e.\ the model for which our main results have been stated in Section~\ref{sec:results}. Thus the underlying infinite graph is now $\mathcal{G}$ with vertex set $\mathcal{V}=\Z$ and edges between $\ell^1$--next neighbors and we look at the induced subgraph with vertices $\mathcal{V}'=\Lambda := [1,L]$. For $\Lambda_0 = [1,\ell]$, $1\le \ell < L$, we consider the decomposition $\mathcal{H}_\Lambda=\mathcal{H}_{\Lambda_0}\otimes \mathcal{H}_{\Lambda_0^c}$ and bipartite entanglement of a normalized state $\psi$ defined by \eqref{entangle}.

For every $X\subseteq \Lambda$ let $\phi_X$ denote the canonical basis function given by
\begin{equation} \label{eq:kronecker}
\phi_X(X')=\begin{cases} &1\quad\mbox{if}\quad X=X'\\
&0\quad\mbox{if}\quad X\neq X'\:.\end{cases}
\end{equation}
which means that $\{\phi_X: X \subseteq \Lambda\}$ is an orthonormal basis of $\mathcal{H}_\Lambda$ (using the identification described in Remark~\ref{ident}).
For $Y\subseteq \Lambda_0$ and $Z\subseteq \Lambda_0^c$ we naturally identify $\phi_Y \otimes \phi_Z = \phi_{Y\cup Z}$.

Let $\psi=\sum_{Y \subseteq \Lambda_0, Z \subseteq \Lambda_0^c}\langle \phi_{Y \cup Z},\psi\rangle \phi_{Y \cup Z}$ be a normalized vector in $\mathcal{H}_\Lambda$ and consider the pure state $\rho:=|\psi\rangle\langle\psi|$. Then the reduced state $\rho_1:=\Tr_{\Lambda_0^c}\rho$ is
\begin{equation} \label{redstate}
\rho_1=\sum_{Y,Y'\subseteq\Lambda_0}\sum_{Z\subseteq \Lambda_0^c}\langle\phi_{Y \cup Z},\psi\rangle
\overline{\langle\phi_{Y' \cup Z},\psi\rangle}\ |\phi_Y\rangle\langle\phi_{Y'}|.
\end{equation}
Recalling that $\mathcal{V}'_N=\{X\subseteq\Lambda: |X|=N\} = \{(x_1,\ldots,x_N) \subseteq \Lambda: x_1 < x_2 < \ldots < x_N\}$, the vertex set of the $N$-th symmetric product of $\mathcal{G}$,  then for $0\le N \le L$, note that
\begin{equation}
{\mathcal H}_\Lambda^N := \mbox{span}\{\phi_X: X \subseteq \Lambda, |X| = N\}=\mbox{span}\{\phi_X: X \in \mathcal{V}'_N\}
\end{equation}
are the $N$-particle ($N$-down-spin) subspaces of $\mathcal{H}_\Lambda$.  Thus $\mathcal{H}_\Lambda = \bigoplus_{N=1}^L \mathcal{H}_\Lambda^N$.
An analogous observation is true for $\mathcal{H}_{\Lambda_0}^N$, where $0\le N \le \ell$.

A normalized $\psi \in \mathcal{H}_\Lambda$ is thus of the form
\begin{equation} \label{psiexp}
\psi = \sum_{N=0}^L a_N\psi_N, \quad \sum_{N=0}^L |a_N|^2=1,
\end{equation}
for an orthonormal system $\{\psi_N\in \mathcal{H}_\Lambda^N, N=0,\ldots,L\}$ in $\mathcal{H}_\Lambda$ and $\psi_N=\sum_{X\subseteq \Lambda, |X|=N} c_X \phi_X $,  $\sum_{X}|c_X|^2=1$.

Let's write $\psi$ as
\begin{equation}
\psi=\Psi\otimes\phi_\emptyset+\hat{\psi},
\end{equation}
where
\begin{equation}\label{def:Psi}
\Psi:=\sum_{N=0}^\ell a_N\sum_{Y\subseteq\Lambda_0; |Y|=N} c_Y\phi_Y,\quad
\hat{\psi}:=\sum_{N=1}^L a_N\sum_{\tiny{
\begin{array}{c}
  Y\subseteq\Lambda_0, Z\subseteq\Lambda_0^c, Z\neq\emptyset; \\
|Y|+|Z|=N
\end{array}
}} c_{Y\cup Z}\phi_{Y}\otimes\phi_Z.
\end{equation}
Since $\Tr_{\Lambda_0^c}|\Psi\otimes \phi_\emptyset\rangle\langle\hat{\psi}|=0$, then
\begin{equation}\label{eq:rho1}
\rho_1=\Tr_{\Lambda_0^c}|\psi\rangle\langle\psi|=|\Psi\rangle\langle\Psi|+\hat{\rho}_1,
\end{equation}
where
\begin{equation}
\hat{\rho}_1:=\Tr_{\Lambda_0^c}|\hat{\psi}\rangle\langle
\hat{\psi}|
=\sum_{N,M=1}^L a_N \overline{a_M}
\sum_{
{\tiny
\begin{array}{ccc}
  Y,Y'\subseteq \Lambda_0, Z\subseteq \Lambda_0^c, Z\neq\emptyset \\
  |Y|+|Z|=N,\ |Y'|+|Z|=M
\end{array}
}
} c_{Y\cup Z} \overline{c_{Y'\cup Z}} |\phi_Y\rangle\langle \phi_{Y'}|.
\end{equation}
We note here that the reduced state $\rho_1$ in (\ref{eq:rho1}) is written
 as a sum of two non-negative operators ($\hat{\rho}_1$ is the partial trace of a non-negative operator).

We will use that, for $0<\alpha<1$, the $\alpha$-R\'{e}nyi entanglement entropy
\begin{equation}
\mathcal{E}_\alpha(\rho) := \frac{1}{1-\alpha}\log \Tr [(\rho_1)^{\alpha}]
\end{equation}
is an upper bound for the von Neumann entanglement entropy, i.e.,
$\mathcal{E}(\rho) = \mathcal{S}(\rho_1) \le  \mathcal{E}_\alpha(\rho)$.
Thus we need to find bounds for $\Tr [(\rho_1)^{\alpha}]$, which reduces to bounds for $\Tr\left[(\hat{\rho}_1)^\alpha\right]$. This follows from
 the quasi-norm
property of $\Tr|\cdot|^\alpha$, see e.g., \cite[Theorem 7.8]{Weidmann},
\begin{equation}\label{eq:bound:Tr:rho1hat}
\Tr[(\rho_1)^\alpha]\leq 2 \Tr\left[|\Psi\rangle\langle\Psi|^\alpha\right]+2\Tr\left[(\hat{\rho}_1)^\alpha\right]\leq 2+2\Tr\left[(\hat{\rho}_1)^\alpha\right],
\end{equation}
where we used that $|\Psi\rangle\langle\Psi|$ has rank (at most) one with norm less than our equal to one.

Next we use Jensen's inequality to get
\begin{eqnarray} \label{1}
\Tr[(\hat{\rho}_1)^\alpha] & = & \sum_{Y \subseteq \Lambda_0} \langle \phi_Y, (\hat{\rho}_1)^{\alpha} \phi_Y \rangle  \leq \sum_{Y\subseteq \Lambda_0}\langle\phi_Y,\hat{\rho}_1\phi_Y\rangle^\alpha \\
& \le & 2+ \sum_{j=1}^{\ell-1} \sum_{Y\subseteq \Lambda_0,\, |Y| =j} \langle\phi_Y,\hat{\rho}_1\phi_Y\rangle^\alpha, \notag
\end{eqnarray}
where the terms corresponding to $Y=\emptyset$ and $Y=\Lambda_0$ were bounded trivially using the fact that $\hat{\rho}_1\leq 1$. (This can be seen  e.g., $\hat{\rho}_1\geq 0$ and $\Tr\hat{\rho}_1\leq 1$.)

To find bounds for $\langle \phi_Y, \hat{\rho}_1 \phi_Y \rangle$, where $Y\subseteq \Lambda_0$ with $1\le |Y| \le \ell-1$, we expand with \eqref{redstate},
\begin{eqnarray} \label{rho1decomp}
\langle \phi_Y, \hat{\rho}_1 \phi_Y \rangle & = & \sum_{N=1}^L |a_N|^2\sum_{Z\subseteq \Lambda_0^c,\, Z\neq\emptyset,\, |Y|+|Z|= N} |\langle\phi_{Y \cup Z},\psi_N\rangle|^2\\
& =& \sum_{k=1}^{L-|Y|} |a_{|Y|+k}|^2 \sum_{Z\subseteq \Lambda_0^c,\, |Z|= k}  |\langle\phi_{Y \cup Z},\psi_{|Y|+k} \rangle|^2 \nonumber \\
& = & \sum_{k=1}^{L-|Y|} |a_{|Y|+k}|^2 \|\chi_{\mathcal{A}_{Y,k}} \psi_{|Y|+k} \|^2 \le \max_{1\le k \le L-|Y|} \|\chi_{\mathcal{A}_{Y,k}} \psi_{|Y|+k} \|^2, \notag
\end{eqnarray}
where
\begin{equation}
\mathcal{A}_{Y,k} := \{ Y \cup Z: Z \subseteq \Lambda_0^c, |Z|=k\}\:.
\end{equation}
Here,  $\chi_{\mathcal{A}_{Y,k}}$ is the multiplication operator by the characteristic function of $\mathcal{A}_{Y,k} \subseteq \mathcal{V}'_{|Y|+k}$ (where, for any $N$, $\{\phi_X\}_{X\subseteq \mathcal{V}'_N}$ is the canonical basis of $\ell^2(\mathcal{V}'_N)$).

In summary, substitute (\ref{rho1decomp}) in  (\ref{1}) then in (\ref{eq:bound:Tr:rho1hat}) to produce the bound
\begin{equation}\label{eq:bound:Tr:rho1hat:main}
\Tr[(\rho_1)^\alpha]\leq
6+2 \sum_{j=1}^{\ell-1} \sum_{Y\subseteq \Lambda_0,\, |Y| =j} \max_{1\leq k\leq L-j}\|\chi_{\mathcal{A}_{Y,k}}\psi_{j+k}\|^{2\alpha}.
\end{equation}

For states as in Theorem~\ref{thm:main} we will now apply Lemma \ref{lemma:decaybounds}, which for the chain takes a particularly convenient form. The reason for this is that for any finite $X\subseteq \Z$ or with $|X|=N$, we have
\begin{equation}
D_{\Z}^N(X)=|\partial X|=2\ cl(X),
\end{equation}
where $cl(X)$ denotes the number of connected components in configuration $X$ (clusters).

In particular, the minimum of $D_{\Z}^N(X)$ is given by $D_{\mathcal{G}_N,\min}=2$, so that the constants $C_3$ and $\mu_3$ in Lemma~\ref{lemma:decaybounds} become $N$-independent. Minimizing configurations $X$ consist of only a single cluster (droplet), i.e., are of the form $X=\{x,x+1,\dots,x+(N-1)\}$, $x\in \Z$.

The energy levels in \eqref{ENk} become $E_{N,k} = (1-\frac{1}{\Delta}) (1+\frac{k}{2})$, are independent of $N$ and determine the threshold energies $E_K = K(1-\frac{1}{\Delta})$ in \eqref{Ksplit} via $k = 2(K-1)$, $K=1,2,\ldots$. Odd values of $k$ are irrelevant in the case of the chain, because $D_{\Z}^N(X)$ can only attain even values.

For any $N,K\in\N$, let us now define the sets
\begin{equation}
\mathcal{V}_{N,K}:=\left\{X\in\mathcal{V}_N: cl(X)\leq K\right\}
\end{equation}
i.e. $\mathcal{V}_{N,K}$ denotes the set of all configurations in $\Lambda$ of \emph{up to} $K$ with exactly $N$ particles in total. Also define $\overline{\mathcal{V}}_{N,K}:=\mathcal{V}_N\setminus\mathcal{V}_{N,K}$. As a special case of Lemma~\ref{lemma:decaybounds} we therefore get

\begin{cor}  \label{cor:chainspecdecay}
For any non-negative background field $V$ on $\Lambda$ and any $1\le N\leq |\Lambda|$, let $H^N_\Lambda(V):=H_\Lambda^N+V_\Lambda^N$. Then, for any $\delta>0$, $K\in \N$ and $\mathcal{A}\subseteq \overline{\mathcal{V}}_{N,K}
$, the following estimate holds:
\begin{equation}
\left\|\chi_\mathcal{A} \chi_{[0,E_{K+1}-\delta]}\left(H_\Lambda^N(V)\right)\right\| \leq C_3 e^{-\mu_3 d_N(\mathcal{A},\mathcal{V}_{N,K})}\:.
\end{equation}
The constants $C_3=C_3(K)$ and $\mu_3=\mu_3(K)$ are given by
\begin{equation} \label{eq:constants2}
C_3(K)=3\sqrt{10}\frac{(K+1)^{3/2}}{\min\{1,\delta^{3/2}\}}\quad\mbox{and}\quad \mu_3(K)=\frac{1}{2}\log\left(1+\frac{\delta\Delta}{4(K+1)}\right)\:.
\end{equation}
\end{cor}

\vspace{.3cm}

We can now combine Corollary~\ref{cor:chainspecdecay} with \eqref{eq:bound:Tr:rho1hat:main} and several lemmas proven in Appendix~\ref{app:aux} into

\begin{prop} \label{prop:summarybound}
Let $\psi$ be as in Theorem~\ref{thm:main}, i.e., $\psi \in R(\chi_{[0,E_{K+1}-\delta]}(H_{\Lambda}(V)))$, $\|\psi\|=1$, $\rho = |\psi \rangle \langle \psi|$, $\rho_1 = \Tr_{\Lambda_0^c} \rho$ and $0<\alpha<1$. Then
\begin{equation} \label{summarybound}
\Tr[(\rho_1)^\alpha] \le 6+ 2\sum_{j=1}^{\ell-1} \sum_{Y\subseteq \Lambda_0,\, |Y| =j} C_3 e^{-2\alpha \mu_3 d_{j+1}(Y \cup \{\ell+1\}, \mathcal{V}_{j+1,K})},
\end{equation}
where $C_3$ and $\mu_3$ are as in Corollary~\ref{cor:chainspecdecay}.
\end{prop}

\begin{proof}
Recall the bound (\ref{eq:bound:Tr:rho1hat:main}) for $\Tr[(\rho_1)^\alpha]$.
That $\psi \in R(\chi_{[0,E_{K+1}-\delta]}(H_\Lambda(V)))$ means that $\psi_{n}  \in R(\chi_{[0,E_{K+1}-\delta]}(H_\Lambda^{n}(V)))$ for all $n$. Therefore, by Corollary~\ref{cor:chainspecdecay},
\begin{equation}
\|\chi_{\mathcal{A}_{Y,k}} \psi_{j+k} \| \le \left\|\chi_{\mathcal{A}_{Y,k}} \chi_{[0,E_{K+1}-\delta]}\left(H_\Lambda^{j+k}(V)\right)\right\| \le C_3 e^{-\mu_3 d_{j+k}(\mathcal{A}_{Y,k},\mathcal{V}_{j+k,K})}.
\end{equation}
Next we use Lemma~\ref{lemma:closestconfig}, which says that
\begin{equation}
d_{j+k}(\mathcal{A}_{Y,k}, \mathcal{V}_{j+k,K}) = d_{j+k}(Y^{(k)}, \mathcal{V}_{j+k,K}),
\end{equation}
where $Y^{(k)} := Y \cup \{\ell+1,\ldots, \ell+k\}$.

Finally, by Lemma~\ref{lemma:justoneparticle}, the minimum of the numbers $d_{j+k}(\mathcal{A}_{Y,k}, \mathcal{V}_{j+k,K})$, $1\le k \le L-j$, is attained for $k=1$.	

Combining all these bounds yields \eqref{summarybound}.
\end{proof}

\section{Proof of Theorem~\ref{thm:main}} \label{sec:proofmain}
The following theorem provides an explicit upper bound for the $\alpha$-R\'{e}nyi entanglement entropy.

\begin{thm}\label{thm:K}Let $\Delta>1$, $\alpha\in(0,1)$, $K\in\mathbb{N}$, and $\delta>0$.  For every normalized vector $\psi$ in the range of $\chi_{[0,E_{K+1}-\delta]}(H_\Lambda(V))$ and $\rho_\psi = |\psi\rangle \langle \psi|$, we have
	\begin{equation} \label{eq:alphaRenyi}
	\mathcal{E}_\alpha(\rho_{\psi})\leq \frac{1}{1-\alpha}
\log\left(\frac{2 C_3\cdot K\cdot C(\alpha)^K}{1-e^{-2\alpha\mu_3}}\ \ell^{2K-1}
+6
\right)
	\end{equation}
where $C(\alpha):=\prod_{j=1}^{\infty}(1-e^{-2\alpha\mu_3 j})^{-2}$
and $C_3$ and $\mu_3$ are as in Corollary~\ref{cor:chainspecdecay}.
\end{thm}

This readily implies Theorem~\ref{thm:main}: For $\psi$ as considered there  we infer from \eqref{eq:alphaRenyi} that
\begin{equation}
\limsup_{\ell\rightarrow \infty}\limsup_{L\rightarrow \infty}\frac{\sup_{\psi}\mathcal{E}_\alpha(\rho_\psi)}{\log\ell} \le \frac{2K-1}{1-\alpha}.
\end{equation}
As $\mathcal{E}(\rho_\psi) \le \mathcal{E}_\alpha(\rho_\psi)$, this implies
\begin{equation}
\limsup_{\ell\rightarrow \infty}\limsup_{L\rightarrow \infty}\frac{\sup_{\psi}\mathcal{E}(\rho_\psi)}{\log\ell} \le \frac{2K-1}{1-\alpha}.
\end{equation}
Letting $\alpha \to 0$ yields \eqref{eq:main}.

\begin{proof} (of Theorem~\ref{thm:K}) Recall that, for $K\in\mathbb{N}$,
\begin{equation} \label{nujK}
\mathcal{V}_{n,K} =\{X\subseteq \Lambda;\ |X|=n \text{ and }\ cl(X)\leq K\}
\end{equation}
are the $n$-particle configurations with at most $K$ clusters.
By Proposition~\ref{prop:summarybound} we need to bound
\begin{equation}\label{eq:first-sum}
\sum_{j=1}^{\ell-1}\sum_{Y\subseteq \Lambda_0, |Y|=j}\exp\left(-\gamma d_{j+1}(Y\cup \{\ell+1\},\mathcal{V}_{j+1,K})\right)
\end{equation}
where $\gamma:=2\mu_3\alpha$. For $k,n\in\mathbb{N}$, define
\begin{equation} \label{nujK=}
\mathcal{V}_{n,k}^=:= \{X\subseteq\Lambda;\ |X|=n \text{ and }\ cl(X)= k\},
\end{equation}
with exactly $k$ clusters. Observe that $\mathcal{V}_{n,k}^= =\emptyset$ when $k>n$.

Moreover, for positive integers $1\leq k\leq K$ and $n\geq k$, define
\begin{equation} \label{XinK}
\Xi_{n,k}=\{X\subseteq \Lambda; \ |X|=n,\  k\in\{1,\ldots,K\} \,\mbox{minimal s.\,t.} \,d_n(X,\mathcal{V}_{n,K})=d_n(X,\mathcal{V}_{n,k}^=)\},
\end{equation}
and note that for each fixed $1\leq j\leq \ell-1$,
\begin{equation}
\{Y\subseteq\Lambda_0;\ |Y|=j\}=\biguplus_{k=1}^{\min\{j+1,K\}}\{Y\subseteq\Lambda_0;\ Y\cup\{\ell+1\}\in\Xi_{j+1,k}\}.
\end{equation}
Hence, the sum (\ref{eq:first-sum}) can be written as
\begin{equation}\label{eq:main1}
\sum_{k=1}^K \sum_{j=\max\{k-1,1\}}^{\ell-1}\sum_{\tiny
\begin{array}{c}
  Y\subseteq\Lambda_0; \\
  Y\cup \{\ell+1\}\in\Xi_{j+1,k}
\end{array}
} \exp\left(-\gamma d_{j+1}(Y \cup \{\ell+1\},\mathcal{V}_{j+1,k}^=)\right).
\end{equation}

To ease notations in the following, we apply the change of coordinates
\begin{equation}
\Lambda \mapsto \ell+1-\Lambda,
\end{equation}
i.e., the vertices of $\Lambda$ are labeled by integers in an increasing order from right to left, such that the number 1 labels the first-from-right vertex of $\Lambda_0$. After this change of variables we can re-express \eqref{eq:main1} by
\begin{equation}\label{eq:main2}
\sum_{k=1}^K \sum_{j=\max\{k-1,1\}}^{\ell-1} \sum_{\tiny
\begin{array}{c}
  Y\subseteq \Lambda_0,\, |Y|=j; \\
  Y^{(0)}\in\Xi_{j+1,k}
\end{array}
} \exp\left(-\gamma d_{j+1}(Y^{(0)},\mathcal{V}_{j+1,k}^=)\right),
\end{equation}
where we write $Y^{(0)} := \{0\} \cup Y$.

For a fixed $Y\subseteq\Lambda_0$ such that $Y^{(0)} \in \Xi_{j+1,k}$, we use Lemma~\ref{lem:closestcluster} and the representation (\ref{E:eq:XXhat}) to find a closest $\hat{Y}^{(0)} \in \mathcal{V}_{j+1,k}^=$ and ``magnets'' $\hat{y}_1, \ldots, \hat{y}_{k}$ with the properties
\begin{itemize}
\item[(i)] $\hat{Y}^{(0)} = \bigcup_{r=1}^{k} \mathcal{C}_{\hat{y}_r}$, where $\mathcal{C}_{\hat{y}_r}$ are ordered and non-touching $m_r$-particle droplets with $\sum_r m_r = j+1$, centered at $\hat{y}_r$ in the sense of \eqref{dropcenter}, i.e.,
\begin{equation}
\mathcal{C}_{\hat{y}_r}=\mathcal{C}^L_{\hat{y}_r}\cup\{\hat{y}_r\}\cup\mathcal{C}^R_{\hat{y}_r}
\end{equation}
where
\begin{equation}\label{eq:C-CL-CR}
\mathcal{C}^L_{\hat{y}_r}:=\{ \hat{y}_r-\lfloor\frac{m_r}{2}\rfloor, \ldots, \hat{y}_r-1\}  \text{ and }  \mathcal{C}^R_{\hat{y}_r}:=\{\hat{y}_r+1, \ldots, \hat{y}_r+
\lfloor\frac{m_r}{2}\rfloor-\delta_{m_r,\text{even}}\} .
\end{equation}
Here for $x\in\mathbb{R}$, $\lfloor x\rfloor$ denotes the greatest integer smaller than or equal to $x$ and
\begin{equation}
\delta_{n,\text{even}}=\left\{
                        \begin{array}{ll}
                          1 & \hbox{if $n$ is even} \\
                          0 & \hbox{if $n$ is odd.}
                        \end{array}
                      \right.
\end{equation}

\item[(ii)] $Y^{(0)} = \bigcup_{r=1}^{k} Y^{(0)}_{\hat{y}_r}$ for ordered $m_r$-particle configurations of the form
\begin{equation}
Y^{(0)}_{\hat{y}_r}=Y_{\hat{y}_r}^L \cup\{\hat{y}_r\} \cup Y_{\hat{y}_r}^R,
\end{equation}
where $|\mathcal{C}^L_{\hat{y}_r}| = |Y_{\hat{y}_r}^L|$, $|\mathcal{C}^R_{\hat{y}_r}| = |Y_{\hat{y}_r}^R|$ for all $r=1,\ldots,k$.
\end{itemize}
Thus we obtain
\begin{equation}\label{eq:Y-Yhat}
d_{j+1}(Y^{(0)},\mathcal{V}_{j+1,k}^=)=d_{j+1}(Y^{(0)},\hat{Y}^{(0)})=
d_{m_1}\left(Y^{(0)}_{\hat{y}_1},\mathcal{C}_{\hat{y}_1}\right)+\sum_{r=2}^{k}d_{m_r}(Y^{(0)}_{\hat{y}_r},\mathcal{C}_{\hat{y}_r})
\end{equation}
for some positive integers $m_1,\ldots,m_{k}$ with $\sum m_r=j+1$. Note here that $Y^{(0)}_{\hat{y}_1}$ starts with $\{0\}$, which here and below is the reason for separating the term $d_{m_1}\left(Y^{(0)}_{\hat{y}_1},\mathcal{C}_{\hat{y}_1}\right)$ from the rest of the summation.
Separating $\{0\} \in Y^{(0)}_{\hat{y}_1}$ from the rest of $Y^{(0)}_{\hat{y}_1}$, this allows to write (\ref{eq:Y-Yhat}) as
\begin{eqnarray}\label{eq:d-left-right}
d_{j+1}(Y^{(0)},\mathcal{V}_{j+1,k}^=)&=&(\hat{y}_1-\lfloor\frac{m_1}{2}\rfloor)+
d_{\lfloor\frac{m_1}{2}\rfloor-1}\left(Y_{\hat{y}_1}^L\setminus\{0\},\mathcal{C}_{\hat{y}_1}^L\setminus\{\hat{y}_1-\lfloor\frac{m_1}{2}\rfloor\}\right)
+
\nonumber\\
&&
+d_{\lfloor\frac{m_1}{2}\rfloor-\delta_{m_1,\text{even}}}(Y_{\hat{y}_r}^R,\mathcal{C}_{\hat{y}_r}^R)
+\sum_{r=2}^{k}
\left(d_{\lfloor\frac{m_r}{2}\rfloor}(Y_{\hat{y}_r}^L,\mathcal{C}_{\hat{y}_r}^L)+
d_{\lfloor\frac{m_r}{2}\rfloor-\delta_{m_r,\text{even}}}(Y_{\hat{y}_r}^R,\mathcal{C}_{\hat{y}_r}^R)\right).
\end{eqnarray}
Here we also use the definition $d_{-1}(\cdot,\cdot)=d_0(\cdot,\cdot):=0$,
and we note here that $m_1=1$ corresponds to $\hat{y}_1=0$ and $Y_{\hat{y}_1}^R=Y_{\hat{y}_1}^L=\emptyset$, this means that all the terms associated with $m_1$ in (\ref{eq:d-left-right}) (the first three terms) are zeros in this case. Thus, we assume that $m_1\geq 2$, and hence $\hat{y}_1\geq 1$.

To avoid the distinction of the first droplet associated with $\hat{y}_1$ in the formulas, we slightly abuse notation and denote $m_1-2$ by $m_1$, then it is easy to check that this means that $\lfloor\frac{m_1}{2}\rfloor-1$ is mapped to $\lfloor\frac{m_1}{2}\rfloor$, i.e., we get
\begin{equation}\label{eq:Y-Yhat2}
d_{j+1}(Y^{(0)},\mathcal{V}_{j+1,k}^=)=(\hat{y}_1-\lfloor\frac{m_1}{2}\rfloor-1)+
\sum_{r=1}^{k}
\left(d_{\lfloor\frac{m_r}{2}\rfloor}(Y_{\hat{y}_r}^L,\mathcal{C}_{\hat{y}_r}^L)+
d_{\lfloor\frac{m_r}{2}\rfloor-\delta_{m_r,\text{even}}}(Y_{\hat{y}_r}^R,\mathcal{C}_{\hat{y}_r}^R)\right)
\end{equation}
where $\mathcal{C}^L_{\hat{y}_r}$ and $\mathcal{C}^R_{\hat{y}_r}$ are given by the formulas (\ref{eq:C-CL-CR}) with $m_1\in\mathbb{N}_0$, $m_2,\ldots,m_{k}\in\mathbb{N}$, with $\sum m_r=(j-1)$.

Using the equality (\ref{eq:Y-Yhat2}), the $Y$-sum in (\ref{eq:main2}) is bounded from above by
\begin{eqnarray}\label{eq:main-long}
\lefteqn{\sum_{\tiny
\begin{array}{c}
  Y\subseteq [1,\ell],\, |Y|=j; \\
  Y^{(0)}\in\Xi_{j+1,k}
\end{array}
} \exp\left(-\gamma d_{j+1}(Y^{(0)},\mathcal{V}_{j+1,k}^=)\right) \le }  \\
\sum_{
\tiny
\begin{array}{c}
  m_1\in\mathbb{N}_0,\\
m_2,\ldots,m_{k}\in\mathbb{N}; \\
  \sum m_r=j-1
\end{array}
}&&
\sum_{\tiny
\begin{array}{c}
\hat{y}_1,\ldots,\hat{y}_k;\\
\lfloor\frac{m_1}{2}\rfloor+1\leq\hat{y}_1<\ldots<\hat{y}_k \le \ell
\end{array}
}
e^{-\gamma(\hat{y}_1-\lfloor\frac{m_1}{2}\rfloor-1)}\times\nonumber\\
&&\times
\prod_{r=1}^{k}\left(
\sum_{\tiny
\begin{array}{c}
  Y^L_r\subseteq(\hat{y}_{r-1},\hat{y}_r)  \\
  |Y^L_r|=\lfloor\frac{m_r}{2}\rfloor
\end{array}
}
e^{-\gamma d_{\lfloor\frac{m_r}{2}\rfloor}(Y^L_r,\mathcal{C}^L_{\hat{y}_r})}
\sum_{\tiny
\begin{array}{c}
Y^R_r\subseteq(\hat{y}_r,\hat{y}_{r+1})\\
|Y^R_r|=\lfloor\frac{m_r}{2}\rfloor-\delta_{m_r,\text{even}}
\end{array}
}e^{-\gamma d_{|Y^R_r|}(Y^R_r,\mathcal{C}^R_{\hat{y}_r})}\right) \notag
\end{eqnarray}
where for convenience, we used $\hat{y}_0$ and $\hat{y}_{k+1}$ to denote $\lfloor\frac{m_1}{2}\rfloor+1$ and $\ell$, respectively. When $\lfloor\frac{m_r}{2}\rfloor$ or $\lfloor\frac{m_r}{2}\rfloor-\delta_{m_r,\text{even}}$ are $0$ or $-1$, the corresponding sums in the last line are interpreted as $1$. That this is an upper bound is due to the fact that we have not required in the last line that $Y_r^R$ lies to the left of $Y_{r+1}^L$.

Next, we find a bound for the two sums inside the big product in (\ref{eq:main-long}). To ease notations, we suppress the subscript $r$ and define $m:=\lfloor\frac{m_r}{2}\rfloor-\delta_{m_r,\text{even}}$ (which is $\ge 1$ in the non-trivial cases).  The second sum inside the product can be bounded as follows
\begin{eqnarray}\label{eq:main-long:3}
\sum_{\tiny
\begin{array}{c}
Y^R_r\subseteq(\hat{y}_r,\hat{y}_{r+1})\\
|Y^R_r|=m
\end{array}
}e^{-\gamma d_{m}(Y^R_r,\mathcal{C}^R_{\hat{y}_r})}
&\leq&
\sum_{
\hat{y}<y_1<\ldots<y_{m}<\infty
}e^{-\gamma \sum_{j=1}^{m}(y_j-(\hat{y}+j))}\\
&=&
\sum_{
0<y_1<\ldots<y_{m}<\infty
}e^{-\gamma \sum_{j=1}^{m}(y_j-j)}\nonumber\\
&=&
\sum_{n_1,n_2,\dots,n_m=0}^\infty (e^{-\gamma m})^{n_1} (e^{-\gamma(m-1)})^{n_2}\cdots (e^{-\gamma})^{n_m} \notag \nonumber \\
&=&
\prod_{j=1}^m \frac{1}{1-\exp(-\gamma j)} \notag
\end{eqnarray}
This increasing product converges (to $C(\alpha)^{\frac{1}{2}}$ in (\ref{eq:alphaRenyi})) by elementary facts. Here, we applied the change of coordinates $y_j\mapsto y_j-\hat{y}$ in the first-to-second line, and $n_j\mapsto y_j-y_{j-1}-1$ for $j=1,\ldots,m$ with $y_0:=0$ in the second-to-third line.

A similar argument yields the same bound for the first sum inside the product, i.e.,
\begin{equation}\label{eq:main-long:2}
\sum_{\tiny
\begin{array}{c}
  Y^L_r\subseteq(\hat{y}_{r-1},\hat{y}_r)  \\
  |Y^L_r|=\lfloor\frac{m_r}{2}\rfloor
\end{array}
}
e^{-\gamma d_{\lfloor\frac{m_r}{2}\rfloor}(Y^L_r,\mathcal{C}^L_{\hat{y}_r})}\leq C(\alpha)^{1/2}.
\end{equation}

The  sum over $\{\hat{y}_r\}_{r=1}^{k}$ in (\ref{eq:main-long}) can be bounded as
\begin{equation}\label{eq:main-long:1}
\sum_{\tiny
\begin{array}{c}
\hat{y}_1,\ldots,\hat{y}_{k};\\
\lfloor\frac{m_1}{2}\rfloor+1\leq\hat{y}_1<\ldots<\hat{y}_{k}\leq\ell
\end{array}
}
e^{-\gamma(\hat{y}_1-\lfloor\frac{m_1}{2}\rfloor-1)}\leq
\sum_{\hat{y}_1=\lfloor\frac{m_1}{2}\rfloor+1}^\infty \sum_{\hat{y}_2,\ldots,\hat{y}_{k}=1}^\ell e^{-\gamma(\hat{y}_1-\lfloor\frac{m_1}{2}\rfloor-1)}=\frac{\ell^{k-1}}{1-e^{-\gamma}}.
\end{equation}
The bounds (\ref{eq:main-long:3}), (\ref{eq:main-long:2}), and (\ref{eq:main-long:1}) allow for bounding (\ref{eq:main-long}) as
\begin{equation}\label{eq:main-long-bound}
\frac{C(\alpha)^{k}\ell^{k-1}}{1-e^{-\gamma}}
\sum_{
\tiny
\begin{array}{c}
  t_1,\ldots,t_{k}\in\mathbb{N}_0 \\
  \sum t_r=j-1
\end{array}
} 1=\frac{C(\alpha)^{k}\ell^{k-1}}{1-e^{-\gamma}}
\binom{(j-1)+k-1}{j-1}
\end{equation}
where we use an elementary fact from combinatorics on multiset coefficients, e.g. \cite[page 38]{Feller}.
Then we observe that
\begin{equation}\label{eq:bound-sum-j}
\sum_{j=\max\{k-1,1\}}^{\ell-1}\binom{j+k-2}{j-1}
\leq
\sum_{m=0}^{\ell-1}\binom{m+k-1}{m}
=
\binom{\ell+k-1}{\ell-1}
=\prod_{j=1}^{k}\frac{\ell+j-1}{j}\leq \ell^k
\end{equation}
where we used the elementary identity
\begin{equation}
\sum_{j=0}^n\binom{j+\ell}{j}=\binom{n+\ell+1}{n}.
\end{equation}

Finally, the entanglement bound (\ref{eq:alphaRenyi}) follows from substituting (\ref{eq:bound-sum-j}) and (\ref{eq:main-long-bound}) in (\ref{eq:main2}), and taking the sum over $k$ using the bound
$\sum_{k=1}^K x^{k-1}\leq K x^{K-1}$, for $x\geq 1$.
\end{proof}

\section{Some thoughts on the (Random) Ising Model} \label{sec:Ising}

In this section we will complement our discussion of the Ising phase of the XXZ chain with some detailed calculations for the Ising limit $\Delta=\infty$ and, in particular, proof Theorems~\ref{thm:detasymp} and \ref{dislowK}.

Besides combinatorial arguments, all we will use as basic ingredients are two elementary facts about the bipartite entanglement of a pure state $\rho = |\psi \rangle \langle \psi|$ for $\psi \in \mathcal{H} = \mathcal{H}_1 \otimes \mathcal{H}_2$ for Hilbert spaces $\mathcal{H}_1$ and $\mathcal{H}_2$:
\begin{itemize}
\item[(i)] If $\psi$ has a decomposition
\begin{equation} \label{Schmidt}
\psi = \sum_{j=1}^N \alpha_j \, e_j \otimes f_j,
\end{equation}
then its (bipartite) von Neumann entanglement entropy satisfies
\begin{equation}  \label{trivbound}
\mathcal{E}(\rho) = \mathcal{S}(\rho_1) = - \tr \rho_1 \log \rho_1 \le \log N,
\end{equation}
where $\rho_1 = \tr_{\mathcal{H}_2} \rho$ is the reduced state. This follows because $\rho_1$ acts non-trivially only on the subspace spanned by the vectors $e_j$. This does not require any assumptions on the vectors $e_j$ and $f_j$.

\item[(ii)] If $\{e_j\}$ and $\{f_j\}$ in \eqref{Schmidt} are orthonormal systems, then the reduced state becomes $\rho_1 = \sum_{j=1}^N |\alpha_j|^2 |e_j\rangle \langle e_j|$. In particular, if $\alpha_j = 1/\sqrt{N}$ for all $j$ so that $\rho_1$ has maximal entropy in span$\{e_j\}$, then
\begin{equation} \label{maxent}
\mathcal{E}(\rho) = \log N.
\end{equation}
\end{itemize}
\subsection{Counting $K$-cluster subsets} \label{sec:clustercount}

To prove Theorem~\ref{thm:detasymp} we will need to know the leading order term as $\ell\to\infty$ of the following two quantities:
Let
\begin{equation} \label{Kclust}
\mathcal{P}_{K,\ell} := |\{Y \subseteq [1,\ell]: cl(Y)=K\}| \text{ for }K\in\mathbb{N}_0\text{ (with $\mathcal{P}_{0,\ell}=1$)},
\end{equation}
and
\begin{equation} \label{Kclustend}
\tilde{\mathcal{P}}_{K,\ell} := |\{Y \subseteq [1,\ell]: \ell \in Y, \, cl(Y)=K\}| \text{ for }K\in\mathbb{N},
\end{equation}
where we recall here that $cl(Y)$ denotes the number of connected components (clusters) of configuration $Y$.
\begin{lemma} \label{lem:clustercount}
For each fixed $K\in \N$ we have
\begin{equation} \label{pasymp}
\mathcal{P}_{K,\ell} = \frac{1}{(2K)!} \ell^{2K} + \mathcal{O}(\ell^{2K-1}) \quad \mbox{as $\ell \to \infty$}
\end{equation}
and
\begin{equation} \label{tpasymp}
\tilde{\mathcal{P}}_{K,\ell} = \frac{1}{(2K-1)!} \ell^{2K-1} + \mathcal{O}(\ell^{2K-2}) \quad \mbox{as $\ell \to \infty$.}
\end{equation}
\end{lemma}

\begin{proof} We start by noting that $\mathcal{P}_{K,\ell}=0$ if $2K-1>\ell$ and proceed by an inductive argument in $K$. By a simple count, the number of subintervals of $[1,\ell]$ (not counting the empty set) is $\mathcal{P}_{1,\ell} = \ell(\ell+1)/2$. Also $\tilde{\mathcal{P}}_{1,\ell} = \ell$.

The inductive argument will be based on
\begin{equation} \label{sump}
\mathcal{P}_{K+1,\ell} = \sum_{r=2}^{\ell +1-2K} \tilde{\mathcal{P}}_{1,r-1} \,\mathcal{P}_{K,\ell-r} = \sum_{r=2}^{\ell +1-2K} (r-1) \,\mathcal{P}_{K,\ell-r}.
\end{equation}
This is seen as follows: The term for $r=2$ in the sum counts the $(K+1)$-cluster sets which contain $1$, do not contain $2$, and contain an arbitrary $K$-cluster subset of $[3,\ell]$. The $r=3$ term gives the number of $(K+1)$-cluster sets whose first cluster ends at $2$, which therefore do not contain $3$, and have $K$ clusters in $[4,\ell]$. Proceeding like this, we get a one-to-one count of all $(K+1)$-cluster subsets of $[1,\ell]$.

Now assume that \eqref{pasymp} holds for a given $K$. Then, by \eqref{sump} and the inductive assumption,
\begin{eqnarray} \label{step1}
\mathcal{P}_{K+1,\ell} & = & \frac{1}{(2K)!} \sum_{r=2}^{\ell+1-2K} (r-1)  (\ell-r)^{2K} + \sum_{r=2}^{\ell+1-2K} (r-1) (\mathcal{P}_{K,\ell-r}-\frac{1}{(2K)!} (\ell-r)^{2K}) \\
& = &  \frac{1}{(2K)!}  \sum_{r=2}^{\ell+1-2K} (r-1) (\ell-r)^{2K} + \mathcal{O}(\ell^{2K+1}), \notag \\
& = & \frac{1}{(2K)!}  \sum_{r=1}^{\ell} r (\ell-r)^{2K} + \mathcal{O}(\ell^{2K+1}). \notag
\end{eqnarray}
Now we can use the integral comparison test to see
\begin{eqnarray} \label{step2}
\sum_{r=1}^\ell r(\ell-r)^{2K} & = & \ell \sum_{s=0}^{\ell-1} s^{2K} - \sum_{s=0}^{\ell-1} s^{2K+1} \\
& = & \ell \left( \frac{\ell^{2K+1}}{2K+1} + \mathcal{O}(\ell^{2K}) \right) - \frac{\ell^{2K+2}}{2K+2} + \mathcal{O}(\ell^{2K+1}) \notag \\
& = & \frac{1}{(2K+1)(2K+2)} \ell^{2K+2} + \mathcal{O}(\ell^{2K+1}). \notag
\end{eqnarray}
Combined, \eqref{step1} and \eqref{step2} yield the inductive step, thus completing the proof of \eqref{pasymp}.

To show \eqref{tpasymp} we use
\begin{equation}
\tilde{\mathcal{P}}_{K,\ell} = \mathcal{P}_{K-1, \ell-2} + \mathcal{P}_{K-1,\ell-3} + \ldots + \mathcal{P}_{K-1, 2K-3},
\end{equation}
which one sees by ``conditioning'' the sets $Y\subseteq [1,\ell]$ with $cl(Y)=K$ and $\ell \in Y$ on the length of their last cluster. From this and \eqref{pasymp} we get
\begin{eqnarray}
\tilde{\mathcal{P}}_{K,\ell} & = & \sum_{j=2}^{\ell-2K+3} \left( \frac{1}{(2(K-1))!} (\ell-j)^{2(K-1)} + \mathcal{O}((\ell-j)^{2K-3}) \right) \\
& = & \frac{1}{(2(K-1))!} \sum_{s=1}^{\ell} s^{2(K-1)} + \mathcal{O}(\ell^{2K-2}) \notag \\
& = & \frac{1}{(2K-1)!} \ell^{2K-1} + \mathcal{O}(\ell^{2K-2}), \notag
\end{eqnarray}
again with the integral comparison test.
\end{proof}

\subsection{Proof of Theorem~\ref{thm:detasymp}}

We introduce, for any $K\in \N$, the collections of sets
\begin{equation}
\mathcal{B}_1 := \{Y \subseteq [1,\ell]: \ell \in Y, \,cl(Y)=K\}
\end{equation}
and
\begin{equation}
\mathcal{B}_2 := \{Y \subseteq [1,\ell]: \,cl(Y) \le K-1\},
\end{equation}
and write $N_{K,\ell} := |\mathcal{B}_1|+|\mathcal{B}_2|$.  Lemma~\ref{lem:clustercount} gives
\begin{equation}\label{eq:NKell}
N_{K,\ell} = \tilde{\mathcal{P}}_{K,\ell} + \mathcal{P}_{0,\ell} + \mathcal{P}_{1,\ell} + \cdots + \mathcal{P}_{K-1,\ell} = \frac{1}{(2K-1)!} \ell^{2K-1} + \mathcal{O}(\ell^{2K-2}).
\end{equation}

\begin{lemma} \label{lemdet}
(a) If $\psi \in R(\chi_{[0,K]}(H_\Lambda^\infty))$, $\|\psi\|=1$, then
\begin{equation} \label{detUBgenK}
\mathcal{E}(\rho_\psi) \le \log (N_{K,\ell}+1).
\end{equation}

(b) If $L$ and $\ell$ are sufficiently large, then there exists a normalized $\psi_0 \in R(\chi_{[0,K]}(H_\Lambda^\infty))$ such that
\begin{equation} \label{genKsat}
\mathcal{E}(\rho_{\psi_0}) = \log (N_{K,\ell}+1).
\end{equation}
\end{lemma}
Using (\ref{eq:NKell}), Theorem~\ref{thm:detasymp} follows from Lemma~\ref{lemdet} and it remains to prove the latter.

To prove part (a) we note that each $\psi$ considered here is of the form
\begin{equation} \label{VgenK}
\psi = \sum_{X\subseteq [1,L]: cl(X) \le K} \alpha_X \phi_X, \quad \sum_X |\alpha_X|^2 = 1.
\end{equation}
If we write $X = Y_X \cup Z_X$ with $Y_X := X \cap [1,\ell]$ and $Z_X := X\cap [\ell+1,L]$, then for each $X$ with $cl(X) \le K$ we must be in exactly one of the following three cases:

\begin{itemize}
\item[(i)] $cl(Y_X) = K$, $\ell \not\in Y_X$. In this case it must hold that $Z_X = \emptyset$,

\item[(ii)] $cl(Y_X) =K$, $\ell \in Y_X$. In this case $Z_X$ must be the empty set or an interval of the form $[\ell+1,\ell+j]$, $1\le j \le L-\ell$.

\item[(iii)] $cl(Y_X) \le K-1$.
\end{itemize}
This means that \eqref{VgenK} has the form
\begin{equation}
\psi = \left( \sum_{Y\subseteq [1,\ell-1]:\  cl(Y) = K} \alpha_Y \phi_Y \right) \otimes \phi_\emptyset + \sum_{Y\in \mathcal{B}_1 \cup \mathcal{B}_2} \phi_{Y} \otimes \Phi_Y,
\end{equation}
where each $\Phi_Y$ is a linear combination of those $\phi_Z$, $Z\subseteq [\ell+1,L]$, such that $X=Y\cup Z$ is either in case (ii) or (iii). This is of the form \eqref{Schmidt} with $N=N_{K,\ell}+1$ and thus proves (a).

(b) Note that $\{\phi_Y: Y \in \mathcal{B}_1 \cup \mathcal{B}_2\}$ is an orthonormal system. Enumerate $\mathcal{B}_1 = \{Y_1,\ldots, Y_{|\mathcal{B}_1|}$\} and $\mathcal{B}_2 = \{Y_{|\mathcal{B}_1| +1}, \ldots, Y_{N_{K,\ell}} \}$. For $j=1,\ldots, N_{K,\ell}$ choose $Z_j := [\ell+1,\ell+j]$ (which is possible for $L$ sufficiently large) and $X_j := Y_j \cup Z_j$. Finally, for $\ell \ge 2K$ one can pick a $Y'\subseteq[1,\ell-1]$ such that $cl(Y')=K$ and choose $X_{N_{K,\ell}+1} = Y'$.

Then, by \eqref{maxent},
\begin{equation}
\psi_0 := \frac{1}{\sqrt{N_{K,\ell}+1}} \sum_{j=1}^{N_{K,\ell}+1} \phi_{X_j} \in R(\chi_{[0,K]}(H_\Lambda^\infty))
\end{equation}
has entanglement $\mathcal{E}(|\psi_0\rangle \langle \psi_0|) = \log(N_{K,\ell}+1)$, proving (b).

\subsection{Proof of Theorem~\ref{dislowK}}

Choose $\delta = \delta_0/K$ and set $p_0 := \mathbb{P}(V_j \le \delta) >0$. Consider the i.i.d.\ Bernoulli random variables $X_j$ with $X_j=1$ if $V_j \le \delta$ and $X_j=0$ if $V_j > \delta$. Let
\begin{equation}
\mathcal{J}:= \sum_{j=1}^\ell X_j = |\{j\in [1,\ell]: V_j \leq \delta\}|.
\end{equation}

Then $\E(\mathcal{J}) = \ell p_0$ and, for any $0<a<p_0$ the Chernoff bound on large deviations (the reason for our assumption on the moment generating function) gives the existence on $\mu>0$, depending on $a$ and $p_0$, such that
\begin{equation} \label{largedev}
\mathbb{P}(\mathcal{J}\le \ell a) \leq e^{-\mu \ell} \quad \mbox{for all $\ell$}.
\end{equation}

For $V \in \Omega_1 := \{V: \mathcal{J} \ge \ell a\}$ we have that the number of (at most) $(K-1)$-element subsets of $A_{V} := \{j\in [1,\ell]: V_j \le \delta\}$ satisfies
\begin{eqnarray} \label{Asubsets}
|\{Y\subseteq A_V: |Y| \le K-1\}| & \ge & {\lfloor \ell a \rfloor \choose K-1} \\
& \ge & \left(\frac{a}{2}\right)^{K-1} \frac{1}{(K-1)!} \ell^{K-1} \notag \\ & =: & C_{K,a} \ell^{K-1} \quad \mbox{for $\ell > 2K/a$.} \notag
\end{eqnarray}

Let $\varepsilon>0$ and choose $\ell_0$ sufficiently large such that $\mathbb{P}(\Omega_1) \ge 1-e^{-\mu \ell_0} \ge 1- \varepsilon/2$ in \eqref{largedev} and that \eqref{Asubsets} holds for $\ell \ge \ell_0$.

Fix one such $\ell$ and consider $L \ge \ell+1$. Let $B_{V,\ell} := \{j \in [\ell+1,L]: V_j \le \delta\}$ and $\Omega_{2,\ell}:= \{V: |B_{V,\ell}| \ge C_{K,a} \ell^{K-1}\}$. Then, again by large deviations, there is $L_0$ sufficiently large (depending on $\ell$) such that for $L\ge L_0$ it holds that
\begin{equation}
\mathbb{P}(\Omega_{2,\ell}) \ge 1- \frac{\varepsilon}{2}.
\end{equation}

Let $V \in \Omega_1 \cap \Omega_{2,\ell}$. Then there exist at least $M := \lfloor C_{K,a} \ell^{K-1} \rfloor$ distinct subsets $Y_j$ of $A_V$, $j=1,2,\ldots$, such that $|Y_j| \le K-1$. There are also at least $\lfloor C_{K,a} \ell^{K-1} \rfloor$ numbers $z_j \in B_{V,\ell}$. Choose
\begin{equation} \label{varphiSchmidt}
\psi = \frac{1}{\sqrt{M}} \sum_{j=1}^M \phi_{Y_j \cup \{z_j\}}.
\end{equation}
By assumption the sets $X_j = Y_j \cup \{z_j\}$ satisfy $cl(X_j) \le K$ and $\sum_{k\in X_j} V_k \le \delta (K-1) + \delta = \delta_0$. Thus, as a linear combination of eigenfunctions to eigenvalues at most $K+\delta_0$, $\psi \in R(\chi_{[0,K+\delta_0]}(H_\Lambda^\infty(V)))$ and $\|\psi\|=1$. Moreover, by \eqref{maxent},
\begin{equation}
\mathcal{E}(\rho_\psi) = \log \lfloor C_{K,a} \ell^{K-1} \rfloor.
\end{equation}
In conclusion
\begin{eqnarray}
\lefteqn{\liminf_{\ell\to\infty} \liminf_{L\to\infty} \frac{1}{\log \ell} \E \left( \sup_{\psi}  \mathcal{E}(\rho_\psi) \right)} \\
& \ge & \liminf_{\ell\to\infty} \liminf_{L\to\infty} \frac{\mathbb{P}(\Omega_1 \cap \Omega_{2,\ell}) \log \lfloor C_{K,a} \ell^{K-1} \rfloor}{\log \ell} \notag \\
& \ge & (1-\varepsilon) (K-1). \notag
\end{eqnarray}
Letting $\varepsilon \to 0$ completes the proof.

\subsection{A conjectured upper bound in the random case} \label{conjecture}

We believe that the following is true for every $K\in \N$:

\vspace{.3cm}

\begin{conjecture} \label{conjupbound}
Let $V_j$ be non-negative i.i.d.\ random variables with distribution $\mu$ such that $0 \in \supp \mu \not= \{0\}$. Let $K\in \N$ and $\delta_0<1$. Then, for the supremum over all normalized $\psi \in R(\chi_{[0,K+\delta_0]}(H_\Lambda^\infty(V)))$,
\begin{equation} \label{eq:disupper2}
\limsup_{\ell\to\infty} \limsup_{L\to\infty} \frac{1}{\log \ell} \,\E \left( \sup_{\psi} \mathcal{E}(\rho_\psi) \right) \le K-1.
\end{equation}
\end{conjecture}

\vspace{.3cm}

This would imply that for every $K\in \N$, $0<\delta_0<1$ and i.i.d.\ random variables which satisfy the assumption of both Theorem~\ref{dislowK} and the Conjecture (for example, if the $V_j$ are uniformly distributed on $[0,1]$),
\begin{equation}
\lim_{\ell\to\infty} \lim_{L\to\infty} \frac{1}{\log \ell} \,\E \left( \sup_{\psi} \mathcal{E}(\rho_\psi) \right) = K-1.
\end{equation}

\vspace{.3cm}

We have proofs of \eqref{eq:disupper2} for $K=1$ (where it is a special case of the Beaud-Warzel result) and for $K=2$. Instead of giving a full proof for these special cases, let us discuss a probabilistic lemma, which suggests that \eqref{eq:disupper2} is true for all $K$, but also explain where a gap remains for $K\ge 3$ in deducing the full Conjecture from this.

Let $\nu_j$, $j=1,\ldots,\ell$, be i.i.d.\ Bernoulli random variables with $\mathbb{P}(\nu_j=0) =p<1$ and $\mathbb{P}(\nu_j=1) =1-p$. We consider random subsets of $[1,\ell]$,
\begin{equation}
I_\ell(\nu) := \{j\in [1,\ell]: \nu_j=0\}
\end{equation}
and, for any $K\in \N$, the number of subsets of $I_\ell(\nu)$ with $K$ connected components,
\begin{equation}
\mathcal{Q}_{K,\ell}(\nu) := |\{ X \subseteq I_\ell(\nu): cl(X)=K\}|.
\end{equation}

\begin{lemma} \label{GellKasymp}
For any $K\in \N$ and $0\le p <1$ we have
\begin{equation} \label{qbound}
\E(\mathcal{Q}_{K,\ell}) = \frac{1}{K!} \left( \frac{p}{1-p} \right)^K \ell^K + \mathcal{O}(\ell^{K-1}) \quad \mbox{as $\ell\to\infty$}.
\end{equation}
\end{lemma}

\begin{proof}
We first prove \eqref{qbound} for $K=1$. For this we condition on the sets
\begin{equation} \label{Odecomp}
\Omega_0 = \{\nu: \nu_1= \ldots =\nu_\ell=0\}, \quad \Omega_n  =  \{\nu: \nu_n = 1, \nu_{n+1} = \ldots = \nu_\ell =0\}, \:n=1,\ldots,\ell.
\end{equation}
Note that $\mathbb{P}(\Omega_0) = p^{\ell}$, $\mathbb{P}(\Omega_n) = (1-p) p^{\ell-n}$ for all $n=1,\ldots,\ell$. Note also that $\mathcal{Q}_{1,\ell}(\nu) = \ell(\ell+1)/2$ for $\nu \in \Omega_0$ and $\mathcal{Q}_{1,\ell}(\nu) = \mathcal{Q}_{1,n-1}(\nu_1,\ldots,\nu_{n-1}) + (\ell-n)(\ell-n+1)/2$ for $\nu \in \Omega_n$, $n=1,\dots,\ell$. This yields
\begin{eqnarray} \label{recform}
\E(\mathcal{Q}_{1,\ell}) & = & \sum_{n=1}^\ell \E(\mathcal{Q}_{1,\ell} \chi_{\Omega_n}) + \E(\mathcal{Q}_{1,\ell} \chi_{\Omega_0}) \\
& = & (1-p) \sum_{n=1}^{\ell} p^{\ell-n} \left( \E(\mathcal{Q}_{1,n-1}) + \frac{ (\ell-n)(\ell-n+1)}{2} \right)  +  \frac{\ell(\ell+1)}{2}  p^\ell, \notag
\end{eqnarray}
setting $\E(\mathcal{Q}_{1,0}):=0$. Using also $\E(\mathcal{Q}_{1,1})=p$, this recursion can be seen to be solved by
\begin{eqnarray} \label{Eq11}
\E (\mathcal{Q}_{1,\ell}) & = & \ell p + (\ell-1) p^2 + (\ell-2) p^3 + \ldots + 2p^{\ell-1} + p^{\ell} \\
& = & \frac{p}{1-p} \ell + \mathcal{O}(1). \notag
\end{eqnarray}

We also introduce the random analogue of \eqref{Kclustend},
\begin{equation}
\tilde{\mathcal{Q}}_{K,\ell}(\nu) := |\{X \subseteq I_\ell(\nu):\,cl(X)=K,\, \ell \in X\}
\end{equation}
and see by a similar argument that
\begin{equation} \label{tEq11}
\E(\tilde{\mathcal{Q}}_{1,\ell}) = \frac{p}{1-p} + \mathcal{O}(p^{\ell}).
\end{equation}

By an analogue of \eqref{sump}, independence and translation invariance,
\begin{equation}
\E(\mathcal{Q}_{K+1,\ell}) = \sum_{r=2}^{\ell+1-2K} \E(\tilde{\mathcal{Q}}_{1,r-1}) \E(\mathcal{Q}_{K,\ell-r}).
\end{equation}
Using \eqref{Eq11} and \eqref{tEq11}, we can use this to prove \eqref{qbound} by induction on $K$.
\end{proof}

Note the crucial difference of this bound to the corresponding deterministic bound \eqref{pasymp}: As soon as $I_\ell(\nu)$ is truly random, i.e., $p<1$, the leading term order in the number of $K$-component subsets drops from $\ell^{2K}$ to $\ell^K$. It is essentially this drop which is the reason for the corresponding drop from $2K-1$ in \eqref{conjdet} to $K-1$ in \eqref{eq:disupper2}. However, our approach to proving this also requires some control of the variance of $\mathcal{Q}_{K,\ell}$. We believe that there is a bound of the form
\begin{equation} \label{Varconj}
\text{Var}(\mathcal{Q}_{K,\ell}) \le C(p,K) \ell^{2K-1}
\end{equation}
for each $K$ and $0\le p<1$, i.e., essentially that the leading terms in $\E(\mathcal{Q}_{K,\ell}^2)$ and $(\E(\mathcal{Q}_{K,\ell}))^2$ cancel. At this point we have what we believe to be an overly complicated proof of this for $K=1$. From this one can indeed conclude that
\eqref{eq:disupper2} holds for $K=2$ (by an argument somewhat reminiscent of the proof of the upper bound in the deterministic setting of Theorem~\ref{thm:detasymp}). We won't write down these proofs for $K=2$ here, because we hold out hope that we can settle \eqref{Varconj} and the full Conjecture~\ref{conjupbound} in the future.

\bigskip
\appendix
\section{Auxiliary Results} \label{app:aux}

Here we collect some technical lemmas and their proofs. We strongly recommend to draw many pictures while reading these proofs.

\subsection{On the distance formula for the case of the chain}
\label{app:distform}

If the underlying graph is the chain or a subinterval thereof, when dealing with its $n$--th symmetric power, it is natural to only consider \emph{ordered} $n$--element subsets of the form $X=\{x_1<x_2<\cdots<x_n\}, Y=\{y_1<y_2<\cdots<y_n\}$. In particular, the ordered labeling used here indeed gives the minimizer in
\begin{equation}  \label{ndist}
d_n(X,Y)=\min_{\pi \in \mathfrak{S}_n} \sum_{j=1}^n |x_j - y_{\pi(j)}|=\sum_{j=1}^n |x_j - y_{j}|\:,
\end{equation}
the general distance formula on symmetric product graphs, see \eqref{distformula}. Indeed, one sees this by the following argument:
\begin{itemize}
\item[(i)] Consider first the case $n=2$. Thus let $x_1<x_2$ and $y_1 < y_2$. It follows by inspection (of the possible relative locations of $x_1$, $x_2$ to $y_1$, $y_2$) that for all possible cases
\begin{equation} \label{A1i}
|x_1 - y_1| + |x_2-y_2| \le |x_1-y_2| + |x_2-y_1|.
\end{equation}

\item[(ii)] In the general case, let $x_1 < x_2 < \ldots < x_n$ and $y_1 < y_2 < \ldots < y_n$ be two ordered configurations. Let $\pi \in \mathfrak{S}_n$ not be the identity. Thus there is at least one $k$ such that $y_{\pi(k+1)} < y_{\pi(k)}$. Therefore, by (i),
\begin{eqnarray}
\sum_j |x_j - y_{\pi(j)}| & = & \sum_{j\not\in \{k,k+1\}} |x_j- y_{\pi(j)}| + |x_k-y_{\pi(k)}| + |x_{k+1}- y_{\pi(k+1)}| \\
& \ge & \sum_{j\not\in \{k,k+1\}} |x_j- y_{\pi(j)}| + |x_k - y_{\pi(k+1)}| + |x_{k+1}- y_{\pi(k)}|. \notag
\end{eqnarray}
After finitely many transpositions of this type one arrives back at the ordered configuration $(y_1,\ldots,y_n)$, proving \eqref{ndist}.
\end{itemize}
\subsection{Closest droplets and $K$-cluster configurations}

We will generally label $n$-particle configurations $W\in \mathcal{V}_n$ in increasing order $W=\{w_1< w_2 < \ldots < w_n\}$.

\begin{lemma} \label{closestdrop}
	Let $W \in \mathcal{V}_n$.
	\begin{itemize}
	\item[(i)] If $n$ is odd, then the unique $n$-droplet $\hat{W}=\{\hat{w}_1<\ldots<\hat{w}_n\}$ closest to $W$ with respect to $d_n(\cdot,\cdot)$ is the droplet which shares its central particle with $W$, i.e., $w_{(n+1)/2} = \hat{w}_{(n+1)/2}$.
	
	\item[(ii)] If $n$ is even, then there are $w_{\frac{n}{2}+1} - w_{\frac{n}{2}}+1$ (the size of the ``central gap'' in $W$ plus one) closest $n$-droplets, with central particles at $w_{\frac{n}{2}} , \ldots, w_{\frac{n}{2}+1}$ (the positions in the central gap and its boundary), respectively.
\end{itemize}
\end{lemma}

\begin{proof}
	This follows from elementary considerations, which can be based on the following facts: If $W \in \mathcal{V}_n$, $V\in \mathcal{V}_{n,1}$ and $V\pm 1 = \{v_1 \pm 1, \ldots, v_n \pm 1\}$ denote the right and left shift of $V$ (if they are still in $\mathcal{V}_n$), then
	\begin{equation}
	|\{j: w_j < v_j\}| < |\{j:w_j \ge v_j\}| \quad \Longrightarrow \quad d_n(V-1,W) > d_n(V,W).
	\end{equation}
	\begin{equation}
	|\{j: w_j < v_j\}| = |\{j:w_j \ge v_j\}| \quad \Longrightarrow \quad d_n(V-1,W) = d_n(V,W).
	\end{equation}
	\begin{equation}
	|\{j: w_j \le v_j\}| > |\{j:w_j > v_j\}| \quad \Longrightarrow \quad d_n(V+1,W) > d_n(V,W).
	\end{equation}
	\begin{equation}
	|\{j: w_j \le v_j\}| = |\{j:w_j > v_j\}| \quad \Longrightarrow \quad d_n(V+1,W) = d_n(V,W).
	\end{equation}
	
	All of these properties follow by comparing how many $v_j$ get closer to $w_j$ or more distant from $w_j$ (in each case by one) under the corresponding shift. The last two properties are ``mirroring'' the first two.
	
	To see (i), one can use the first and third property to see that starting from the droplet which shares the same center particle with $W$ will always increase the $d_n$-distance. For the even case (ii), one sees the additional degeneracy via the second and fourth property, before the distance starts increasing under further shifting.
\end{proof}

We say that an $n$-particle droplet $\hat{X}$ is \textit{centered} at $\hat{x}$, or $\hat{x}$ is a \textit{central particle} of the droplet $\hat{X}$, if it is of the form:
\begin{equation} \label{dropcenter}
\hat{X}=\{\hat{x}-\lfloor\frac{n}{2}\rfloor,
    \ldots,\hat{x}-1,\hat{x},\hat{x}+1,\ldots,
    \hat{x}+\lfloor\frac{n}{2}\rfloor-\delta_{n,\text{even}}\}
\end{equation}
where
\begin{equation}
\delta_{n,\text{even}}=\left\{
                        \begin{array}{ll}
                          1 & \hbox{if $n$ is even} \\
                          0 & \hbox{if $n$ is odd.}
                        \end{array}
                      \right.
\end{equation}
This means that $\hat{x}$ is the middle particle of $\hat{X}$ if $n$ is odd, and for technical reasons we choose the define the central particle to be the bigger between the two middle particles when $n$ is even.

Recall the definitions of the sets $\mathcal{V}_{n,k}$, $\mathcal{V}^=_{n,k}$ and $\Xi_{n,k}$ in \eqref{nujK}, \eqref{nujK=} and \eqref{XinK}.

\begin{lemma} \label{lem:closestcluster}
Let $X = \{x_1< \ldots < x_n\} \in\Xi_{n,k}$.  Then there exists $\hat{X} = \{ \hat{x}_1 < \ldots < \hat{x}_n\} \in\mathcal{V}_{n,k}^=$ with $d_n(X,\mathcal{V}_{n,k}^=) = d_n(X,\hat{X})$ and $x_{j_r} = \hat{x}_{j_r}$, $r=1,\ldots, k$, where $\hat{x}_{j_r}$ is the central particle of the $r$-th cluster in $\hat{X}$ (written in increasing order).
\end{lemma}

We note that one may think of the $x_{j_r}$ as ``magnets'' within $X$. The closest $k$-cluster $\hat{X}$ is found by moving the other particles in $X$ into clusters centered at the $x_{j_r}$.

\begin{proof}
By definition there is an $\tilde{X}\in\mathcal{V}_{n,k}^=$ such that
\begin{equation}\label{eq:d}
d_n(X,\tilde{X})=\min_{Y\in\mathcal{V}_{n,k}^=}d_n(X,Y).
\end{equation}
We will show that $\tilde{X}$ is of the form required for $\hat{X}$ or can be slightly modified to yield the required form.

There exist $m_1,m_2,\ldots,m_{k}\in\mathbb{N}$ with $\sum m_j=n$ such that $\tilde{X}$ can be written as a union of (non-touching) $m_r$-particle droplets $\tilde{X}_{r}$, i.e.,
\begin{equation}
\tilde{X}=\bigcup_{r=1}^k \tilde{X}_r,
\end{equation}
and we may choose the droplets $\tilde{X}_r$ as ordered in the sense that
\begin{equation}
\max \tilde{X}_r<\min \tilde{X}_{r+1}-1, \text{ for all } 1\leq r\leq k-1.
\end{equation}
We consider a corresponding partition of $X$ into $k$ ordered subsets $\{X_r\}_{r=1}^k$ each containing $m_r$ particles from $X$. In particular, for each $r\in\{1,\ldots,k\}$, we define $X_r\subseteq X$  as follows
\begin{equation}
X_1:=\{x_1,\ldots,x_{m_1}\},\ X_2:=\{x_{m_1+1},\ldots,x_{m_1+m_2}\}, \ldots, X_k:=\{x_{n-m_r-1},\ldots,x_n\}.
\end{equation}
Then, by Appendix~\ref{app:distform},
\begin{equation}\label{eq:cluster-droplet}
d_n(X,\tilde{X})= \sum_{r=1}^k d_{m_r}(X_r,\tilde{X}_r).
\end{equation}
We have $\min X_r \le \min \tilde{X}_r$ and $\max X_r \ge \max \tilde{X}_r$ for all $r$ (otherwise the distance of $\tilde{X}$ to $X$ could be reduced by shifting some of its clusters, contradicting the minimality property of $\tilde{X}$).

For each $r\in \{1,\ldots,k\}$ we consider two cases and in each case argue with Lemma \ref{closestdrop} and the minimality property of $\tilde{X}$:
\begin{itemize}
\item[(i)] $m_r$ is odd. In this case $X_j$ and $\tilde{X}_j$ must have their central particle in the same position.
\item[(ii)] $m_r$ is even. Then either $\tilde{X}_r$ or a suitable right-shift of $\tilde{X}_r$ will have the same central particle as $X_r$ (without change of the distance to $X_r$).
\end{itemize}
Accordingly modifying $\tilde{X}$ we get $\hat{X}$ of the required form.
\end{proof}

We label the droplets of $\hat{X}$ and their corresponding configurations in $X\in\mathcal{V}_{n,k}^=$ using these central particles, i.e., we write
\begin{equation}\label{E:eq:XXhat}
X=\bigcup_{r=1}^k X_{\hat{x}_r}, \text{ and } \hat{X}=\bigcup_{r=1}^k \mathcal{C}_{\hat{x}_r},\ \mathcal{C}_{\hat{x}_r}\in\mathcal{V}_{m_r,1}^=, \text{ and }\sum_{r=1}^k m_r=n.
\end{equation}

\subsection{The distances $d_{|Y|+k}(\mathcal{A}_{Y,k}, \mathcal{V}_{|Y|+k,K})$}

Finally, we prove the properties of the distances $d_{|Y|+k}(\mathcal{A}_{Y,k}, \mathcal{V}_{|Y|+k,K})$ which we have used in the proof of Proposition~\ref{prop:summarybound} above.

\begin{lemma} \label{lemma:closestconfig}
	To any $Y\subseteq \Lambda_0$ with $1\le |Y| \le \ell-1$ and $1\le k \le L-|Y|$, let
	\begin{equation}
	Y^{(k)} := Y \cup \{\ell+1, \ell+2, \ldots, \ell+k\} \in \mathcal{A}_{Y,k}.
	\end{equation}
	Then, for any $K\in \N$,
	\begin{equation} \label{distmin}
	d_{|Y|+k}(\mathcal{A}_{Y,k}, \mathcal{V}_{|Y|+k,K}) = d_{|Y|+k}(Y^{(k)}, \mathcal{V}_{|Y|+k,K}).
	\end{equation}
\end{lemma}

\begin{proof}[Proof of Lemma \ref{lemma:closestconfig}]
Let $V \in \mathcal{A}_{Y,k}$ and $W\in \mathcal{V}_{|Y|+k,K}$ a distance minimizing pair, i.e., $d_{|Y|+k}(V,W) = d_{|Y|+k}(\mathcal{A}_{Y,k}, \mathcal{V}_{|Y|+k,K})$. If $V\cap \{\ell+1,\ldots,\ell+k\} \not= \{\ell+1,\ldots,\ell+k\}$, then one can iteratively use Lemma~\ref{lem:onestep} below to move the right-most components of $V$ to the left, until one arrives at $Y^{(k)}$, without increasing the distance to $\mathcal{V}_{|Y|+k,K}$. This gives \eqref{distmin}.
\end{proof}

\begin{lemma} \label{lem:onestep}
Let $V = (v_1,\ldots, v_n) \in \mathcal{V}_n \setminus \mathcal{V}_{n,1}$ and $(v_{j+1},\ldots,v_n)$ its right-most connected component. Let $V' = (v_1,\ldots,v_j, v_{j+1}-1,\ldots, v_n-1)$. Then, for any $K\in \N$,
\begin{equation} \label{dnmon}
d_n(V', \mathcal{V}_{n,K}) \le d_n(V,\mathcal{V}_{n,K}).
\end{equation}
\end{lemma}

\begin{proof} Choose $Z = (z_1,\ldots,z_n) \in \mathcal{V}_{n,K}$ with minimal distance to $V$. It is easy to see that $v_1 \le z_1$ and $z_n \le v_n$. Write $Z= Z_1 \cup Z_2$, where $Z_2= (z_{k+1},\ldots,z_n)$ is the right-most connected component of $Z$. Let $Z' := Z_1 \cup (Z_2-1) \in \mathcal{V}_{n,K}$. One can see that $k\le j$ ($k>j$ would be a contradiction to $Z$ having minimal distance to $V$, as in this case one could move $z_{j+1}, \ldots, z_k$ to the $k-j$ sites directly to the left of $z_{k+1}$, reducing the distance to $V$ without increasing the number of clusters of $Z$).

Now there are two cases: If $z_n<v_n$, then $d_n(V',Z) < d_n(V,Z)$. In case $z_n=v_n$, using $k\le j$, it follows that $d_n(V',Z') \le d_n(V,Z)$. Thus in both cases one gets \eqref{dnmon}.

\end{proof}

\begin{lemma} \label{lemma:justoneparticle}
	The minimum of the numbers $d_{|Y|+k}(\mathcal{A}_{Y,k}, \mathcal{V}_{|Y|+k,K})$, $1\le k \le L-|Y|$, is attained for $k=1$.
\end{lemma}

\begin{proof}
	By Lemma \ref{lemma:closestconfig}, we have to consider $d_{|Y|+k}(Y^{(k)}, \mathcal{V}_{|Y|+k,K})$, $1\le k \le L-|Y|$. It is geometrically quite evident that these numbers are non-decreasing in $k$: If one can move $Y^{(k+1)}$ to a set $Z$ with no more than $K$ clusters in $s$ steps, then in doing so one has also moved $Y^{(k)}$ (appearing as a ``restriction'' of $Y^{(k+1)}$) to a $K$-cluster set ($Z$ without its right-most element) in at most $s$ steps.
\end{proof}

\end{document}